\documentclass[letterpaper,onecolumn,11pt,accepted=2021-08-30,allowtoday]{quantumarticle}

\pdfoutput=1
\usepackage[T1]{fontenc}
\usepackage{authblk}
\usepackage{amsfonts}
\usepackage{amsmath}
\usepackage{amssymb}
\usepackage{amsthm}
\usepackage{graphicx}
\usepackage[colorlinks=true,plainpages=false,pdfpagelabels]{hyperref}
\usepackage{color}
\usepackage[dvipsnames]{xcolor}
\usepackage{mathtools}
\usepackage{bbm}
\usepackage{float}
\usepackage{array}
\usepackage{verbatim}
\usepackage{enumitem}
\usepackage{xfrac}
\usepackage[title,titletoc]{appendix}

\usepackage{anyfontsize}  
\usepackage{microtype} 

\allowdisplaybreaks

\makeatletter
\g@addto@macro\bfseries{\boldmath}
\makeatother

\newcommand{\bra}[1]{\langle #1|}
\newcommand{\ket}[1]{|#1\rangle}

\newcommand{\abs}[1]{\left|#1\right|}

\newcommand{\Tr}{\operatorname{Tr}}

\numberwithin{equation}{section}  

\addtocontents{toc}{\protect\setcounter{tocdepth}{2}}

\RequirePackage[backend=biber,style=phys,eprint=true,arxiv=abs,doi=false,articletitle=true,biblabel=brackets,chaptertitle=false,pageranges=false,bibencoding=utf8]{biblatex}
\appto{\bibsetup}{\sloppy} 

\DeclarePairedDelimiter{\floor}{\lfloor}{\rfloor}

\DeclareMathOperator*{\argmax}{arg\,max}

\theoremstyle{definition}
\newtheorem{theorem}{Theorem}[section]
\newtheorem{lemma}[theorem]{Lemma}
\newtheorem{corollary}[theorem]{Corollary}
\newtheorem{definition}[theorem]{Definition}

\newtheorem{proposition}[theorem]{Proposition}
\newtheorem{remark}[theorem]{Remark}

\newcommand{\defqed}{$\blacktriangleleft$}
\newcommand{\defqedspec}{\blacktriangleleft}

\begin{document}

\title{Policies for elementary links in a quantum network}

\author{Sumeet Khatri}
\affiliation{Hearne Institute for Theoretical Physics, Department of Physics and Astronomy, and Center for Computation and Technology, Louisiana State University, Baton Rouge, Louisiana, 70803, USA}
\orcid{0000-0002-9858-0511}
\email{khatri6000@gmail.com}
\homepage{http://sumeetkhatri.com}

\date{\today}

\maketitle

\begin{abstract}

	Distributing entanglement over long distances is one of the central tasks in quantum networks. An important problem, especially for near-term quantum networks, is to develop optimal entanglement distribution protocols that take into account the limitations of current and near-term hardware, such as quantum memories with limited coherence time. We address this problem by initiating the study of quantum network protocols for entanglement distribution using the theory of decision processes, such that optimal protocols (referred to as \textit{policies} in the context of decision processes) can be found using dynamic programming or reinforcement learning algorithms. As a first step, in this work we focus exclusively on the elementary link level. We start by defining a quantum decision process for elementary links, along with figures of merit for evaluating policies. We then provide two algorithms for determining policies, one of which we prove to be optimal (with respect to fidelity and success probability) among all policies. Then we show that the previously-studied memory-cutoff protocol can be phrased as a policy within our decision process framework, allowing us to obtain several new fundamental results about it. The conceptual developments and results of this work pave the way for the systematic study of the fundamental limitations of near-term quantum networks, and the requirements for physically realizing them.
	
\end{abstract}

\tableofcontents

\section{Introduction}

	The quantum internet \cite{Kim08,Sim17,Cast18,WEH18,Dowling_book2} is one of the frontiers of quantum information science. It has the potential to revolutionize the way we communicate and do other tasks, and it will allow for tasks that are not possible using the current, classical internet alone, such as quantum teleportation \cite{BBC+93,Vaidman94,BFK00}, quantum key distribution \cite{BB84,Eke91,GRG+02,SBPC+09}, quantum clock synchronization \cite{JADW00,Preskill00,UD02,ITDB17}, distributed quantum computation \cite{CEHM99}, and distributed quantum metrology and sensing \cite{KKB+14,DRC17,ZZS18,EFM+18,PKD18,XZCZ19}. The backbone of a quantum internet is entanglement distributed globally in order to allow for such novel applications to be performed over long distances. Consequently, long-range entanglement distribution is one of the main problems in quantum networks. 
	
	Most of the aforementioned applications are beyond the reach of current and near-term quantum technologies. Consequently, we are currently in the era of so-called \textit{near-term quantum networks}, which are characterized by the following elements \cite{WEH18}: Small number of nodes; imperfect sources of entanglement; non-deterministic elementary link generation and entanglement swapping; imperfect measurements and gate operations; quantum memories with short coherence times; no (or limited) entanglement distillation/error correction. The most prominent applications of these near-term quantum networks are quantum teleportation and quantum key distribution. In fact, several experiments have already realized these applications on relatively small scales \cite{MMO+07,MOH+09,PPA+09,CWL+10,MP10,SLB+11,SFI+11,Ritter+12,HKO+12,BHP+13,WCY+14,DSG+16,SSH+17,Kalb+17,YCL17,Hump+18,BLL+18,ZXCPP18,SNN+20,PHB+21}. 
	
	What are the requirements for physically realizing near-term quantum networks? More generally, what are the limitations of near-term quantum networks? Although several software tools for simulating quantum networks have been released in order to probe these questions \cite{MMG+15,DW18,Bart18,Mat19,DNZB20,WKC+20,CKD+20}, it is of interest to develop a formal and systematic theoretical framework for entanglement distribution protocols in near-term quantum networks that can allow us to address these questions in full generality. Such a theoretical framework, which is currently lacking (see Appendix~\ref{app-related_work} for a review of prior theoretical work on quantum networks), should incorporate both the limitations of near-term quantum technologies and be general enough to allow for optimization of protocol parameters. It should also allow us to answer the following basic questions for arbitrary protocols: (1) What is the quantum state of the network? (2) What is the fidelity of the quantum state of the network with respect to a given target state? (3) What protocol is optimal with respect to fidelity (or some other figure of merit)? The purpose of this work is to answer these questions at the level of elementary links in a quantum network, as a first step towards the development of a general framework for practical quantum network protocols. We do so by introducing a general framework for analyzing elementary links in a quantum network based on quantum decision processes; see Figure~\ref{fig-net_agent_env_0}. In a decision process \cite{Put14_book}, an agent interacts with its environment through actions, and it receives rewards from the environment based on these actions. The goal of the agent is to devise a policy (a sequence of actions) that maximizes its expected total reward.

	\begin{figure}
		\centering
		\includegraphics[width=0.80\textwidth]{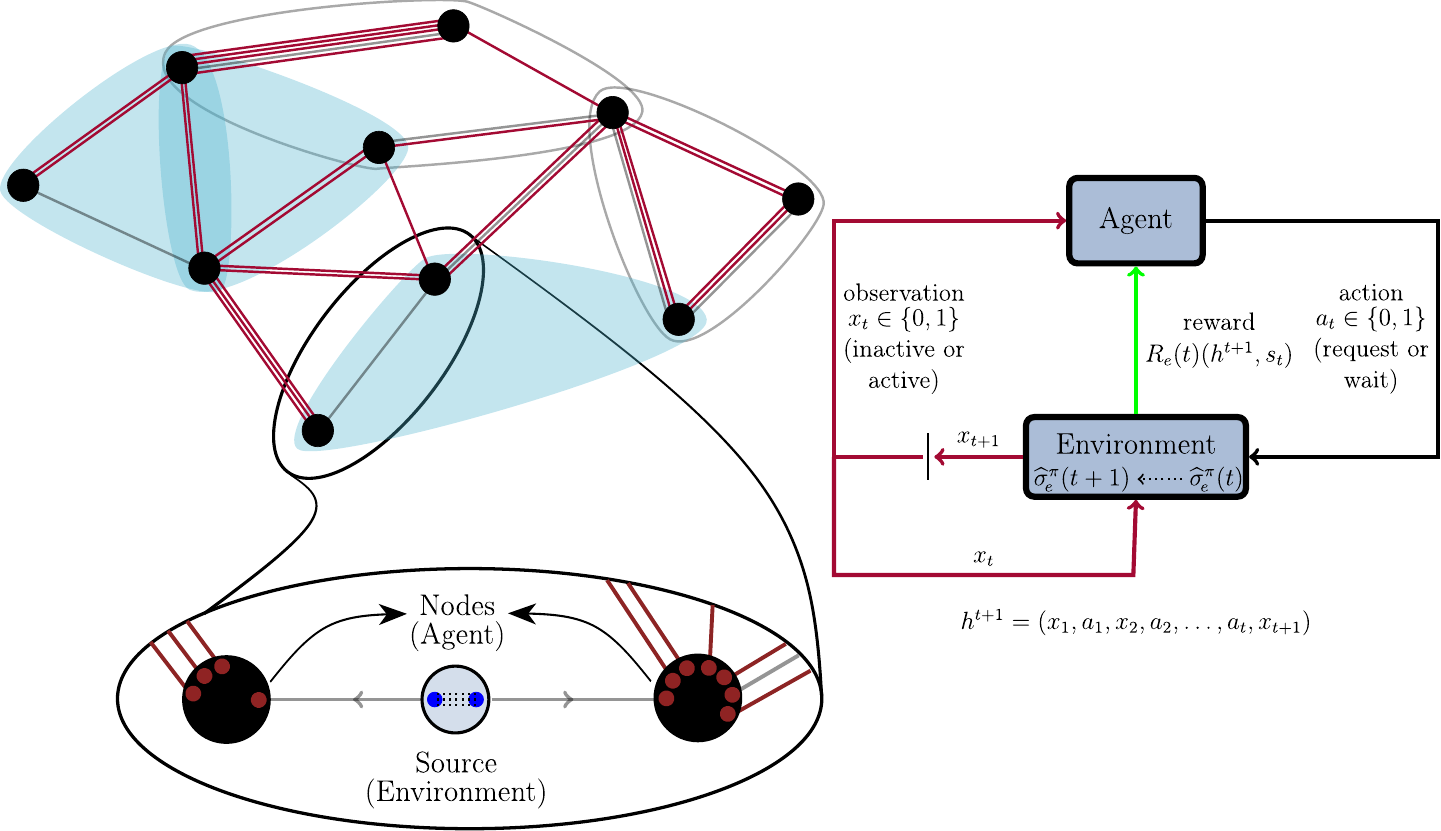}
		\caption{(Left) Model of elementary link generation considered in this work (see Section~\ref{sec-ent_dist_task} and Appendix~\ref{sec-elem_link_generation} for details). We associate every quantum network with a graph $G=(V,E)$, in which the vertices $V$ represent the network nodes and the edges $E$ represent quantum channels, which we refer to as \textit{elementary links}. We associate a source station to the elementary link corresponding to $e\in E$, which distributes entangled states to the nodes belonging to $e$. (Right) In order to analyze an elementary link with respect to time, we define a quantum decision process; see the beginning of Section~\ref{sec-elem_link_QDP}, and Appendix~\ref{app-elem_link_QDP}, for details.} 
\label{fig-net_agent_env_0}
	\end{figure}
	
	To see why decision processes are a natural way to describe protocols in a quantum network, consider a very simple example. Consider three nodes labeled $A$, $B$, and $C$. Quantum channels connect $A$ and $C$ and $B$ and $C$, and the goal is to create entanglement between $A$ and $B$. The usual protocol to achieve this goal is to first create entanglement between $A$ and $C$ and $B$ and $C$, and then to perform entanglement swapping at $C$. However, because entanglement creation is typically non-deterministic with near-term quantum technologies, it is possible that, e.g., the entanglement between $A$ and $C$ is created first. With near-term quantum memories, storing this entangled state for too long while entanglement between $B$ and $C$ is being created can lead to too much decoherence, rendering the final, swapped state between $A$ and $B$ useless. A decision must therefore be made to either wait (i.e., keep the entanglement between $A$ and $C$ in quantum memory) or to discard the entanglement shared by $A$ and $C$ and create the entanglement again. The framework of decision processes provides us with the language and mathematical tools needed to address this problem for arbitrary networks, not just a network of three nodes. 
	
	In addition to being natural, one of the other advantages of the approach taken in this work is that optimal protocols can be discovered using reinforcement learning algorithms. This is due to the fact that decision processes form the theoretical foundation for reinforcement learning \cite{Sut18_book} and artificial intelligence \cite{RN09_book}. (See \cite{WMDB19} for related work on machine learning for quantum communication.) Another advantage of our approach is that, even though reinforcement learning techniques cannot always be applied efficiently to large-scale problems, decision processes provide us with a systematic framework for combining optimal small-scale protocols in order to create large-scale protocols; see \cite{JTKL07} for similar ideas. The framework introduced in this work can also be extended to allow for a systematic consideration of agents with local and global knowledge of the network, as well as agents that are independent and/or cooperate with each other. These extensions are interesting directions for future work, and they will lead to a more complete theory of practical quantum network protocols. This work represents the starting point towards this ultimate goal.

\section{The entanglement distribution task}\label{sec-ent_dist_task}

	In the left panel of Figure~\ref{fig-net_agent_env_0}, we illustrate an arbitrary quantum network using its corresponding hypergraph $G=(V,E)$. Specifically, $G$ corresponds to the physical layout of the network, which we assume to be fixed. The nodes of the network are associated to the vertices $V$ of the graph, and quantum channels physically connecting the nodes in the network are represented in the graph by gray edges, belonging to the set $E$, connecting the corresponding vertices. We refer to these edges as \textit{elementary links}. The quantum channels are used to distribute entangled states to the nodes of an elementary link. For two-node elementary links, the associated quantum channels are used to distribute bipartite entangled states; for elementary links with three or more nodes, the associated quantum channels are used to distribute multipartite entangled states. When an entangled state is distributed successfully along an elementary link, we color the corresponding edge red (in the case of two-node elementary links) or blue (in the case of elementary links with three or more nodes), and we refer to the elementary link as an \textit{active elementary link}. The task of \textit{entanglement distribution} is to use the active elementary links to create \textit{virtual links}. A virtual link corresponds to an entangled state shared by nodes that are not physically connected. Entanglement distribution protocols can be described in terms of graph transformations, as done in \cite{SMI+17,CRDW19}, which take the graph $G$ and transform it into a target graph $G_{\text{target}}$, whose edges contain a subset of the elementary links in $G$ along with a desired set of virtual links. The quantum states corresponding to the elementary and virtual links must be close (in terms of, e.g., fidelity) to a target quantum state. Typically, in the bipartite case, the target is a Bell state, while in the multipartite case the target is a GHZ state.
	
	A basic example of an entanglement distribution protocol, and the one that we consider in this work, consists of first generating elementary links, and then performing joining protocols, such as entanglement swapping, to create virtual links. If the active elementary links obtained after the first step, along with their fidelities to the target states, do not allow for the target graph to be created---for example, some of the required elementary link attempts might have failed---then it makes sense to retry the elementary link generation for the ones that failed. For the ones that succeeded, it might make sense to keep the quantum states in memory rather than discard the states, request new ones, and risk some of these new attempts failing. This sequence of decisions at every time step defines a \textit{policy}. Some questions then naturally arise:
	\begin{enumerate}
		\item Given a policy for an elementary link, what is the quantum state of the elementary link as a function of time? We address this question in Section~\ref{sec-q_state_elem_link}.
		
		\item Given a policy for an elementary link, what is the expected fidelity of the quantum state of the elementary link with respect to the target state, as a function of time, and what is the probability that the elementary link is active after a given number of time steps? We address this question in Section~\ref{sec-figures_of_merit_general}.
		
		\item What is the (optimal) sequence of actions (i.e., the optimal policy) that should be performed for every elementary link, as a function of time? We address this question in Section~\ref{sec-policy_opt}.
	\end{enumerate}
	These are the main questions that we address in this work, and we do so using the theory of quantum decision processes.

\section{Quantum decision process for elementary links}\label{sec-elem_link_QDP}

	Let us now illustrate how the three questions posed at the end of the previous section can be answered using quantum decision processes. The basic idea is illustrated in Figure~\ref{fig-net_agent_env_0}, and we present the formal definition of the quantum decision process in Appendix~\ref{app-elem_link_QDP}. Let $G=(V,E)$ be the graph corresponding to the elementary links of a quantum network. To every edge $e\in E$ of the graph we associate an independent agent. The agent should be thought of as a collection of (classical) devices located at the nodes corresponding to the edge, which can communicate with each other and thus operate as a single entity. The environment associated with the agent is the collection of quantum systems distributed to the corresponding nodes by a source station.
	
	Now, at $t=0$, an attempt is made to generate entanglement along an elementary link corresponding to $e\in E$. This means that the source station associated with $e$ prepares a multipartite entangled state and sends the corresponding quantum systems to the nodes of $e$. There are two key elements of this elementary link generation process (see Appendix~\ref{sec-elem_link_generation} for details):
	\begin{itemize}
		\item The \textit{elementary link generation success probability} $p_e\in[0,1]$, which is a function of the transmissivity of the transmission medium (in the case of photonic implementations) and parameters that quantify imperfections in the local gates and measurement devices.
		
		\item Depending on success or failure of elementary link generation, the nodes of $e$ can be in one of two quantum states, by definition: $\rho_e^0$ in the case of success, and $\tau_e^{\varnothing}$ in the case of failure. In the case of success, the quantum state is held in quantum memories at the nodes. Then, the quantum state after $m$ time steps in the local quantum memories is given by
			\begin{equation}\label{eq-elem_link_state_in_mem}
				\rho_{e}(m)\coloneqq\mathcal{N}_{e}^{\circ m}(\rho_{e}^0),
			\end{equation}
			where $\mathcal{N}_{e}^{\circ m}=\mathcal{N}_{e}\circ\mathcal{N}_{e}\circ\dotsb\circ\mathcal{N}_{e}$ ($m$ times), and $\mathcal{N}_e$ describes the noise processes of the local quantum memories; see Appendix~\ref{sec-elem_link_generation} for details.
	\end{itemize}
	
	If the elementary link generation succeeds at time $t=0$, then at time step $t=1$ the agent might decide to keep the quantum state currently in memory; if it fails, then at time step $t=1$ the agent might decide to perform the elementary link generation again. In general, then, at every time step $t\geq 1$, the agent associated with $e\in E$ can perform two actions: ``wait'' (i.e., keep the entangled state currently in quantum memory), or ``request'' (discard the entangled state currently in quantum memory and request a new one from the source). This choice of action can be random, so we define \textit{action} random variables $A_{e}(t)$ taking two values: 0 for ``wait'' and 1 for ``request''. Based on the agent's choice, the distribution of the quantum systems from the source station to the nodes probabilistically succeeds or fails (we expand on this in Appendix~\ref{sec-elem_link_generation}). We define \textit{elementary link status} random variables $X_{e}(t)$ to indicate the outcomes: 0 for failure and 1 for success. If $X_e(t)=0$, then the elementary link is considered \textit{inactive}, and it is considered \textit{active} if $X_e(t)=1$. The \textit{history} of the agent is then defined as the sequence $H_{e}(t)\coloneqq(X_{e}(1),A_{e}(1),X_{e}(2),A_{e}(2),\dotsc,A_{e}(t-1),X_{e}(t))$ for all $t\geq 1$, with $H_{e}(1)=X_{e}(1)$. Every realization of the history is a sequence of the form $h^t=(x_1,a_1,x_2,a_2,\dotsc,a_{t-1},x_t)$, with $x_j\in\{0,1\}$ for all $1\leq j\leq t$ and $a_j\in\{0,1\}$ for all $1\leq j\leq t-1$. Note that $\{0,1\}^{2t-1}$ is the set of all histories up to time $t$. We then think of $X_e(j)$ and $A_e(j)$ as functions such that $X_e(j)(h^t)=x_j$ for all $1\leq j\leq t$ and $A_e(j)(h^t)=a_j$ for all $1\leq j\leq t-1$ and $h^t\in\{0,1\}^{2t-1}$. We use $h_j^t\coloneqq(x_1,a_1,\dotsc,a_{j-1},x_j)$ to denote the history $h^t=(x_1,a_1,\dotsc,a_{t-1},x_t)$ up to the $j^{\text{th}}$ time step for all $1\leq j\leq t$, with $h_1^t=x_1$.
	
	Now, because the actions of waiting and requesting can be random, we define the random variable $M_e(t)$ to be the amount of time the quantum state of the elementary link corresponding to $e$ is held in memory. It satisfies the recursion relation
	\begin{equation}\label{eq-mem_time_def1}
		M_{e}(t)=\left\{\begin{array}{l l} M_{e}(t-1)+X_{e}(t) & \text{if }A_{e}(t-1)=0,\\ X_{e}(t)-1 & \text{if }A_{e}(t-1)=1, \end{array}\right.
	\end{equation}
	where $M_{e}(0)\equiv -1$ and $A_{e}(0)\equiv 1$. Intuitively, the quantity $M_{e}(t)$ is the number of consecutive time steps up to the $t^{\text{th}}$ time step that the action ``wait'' is performed since the most recent ``request'' action. The value $M_{e}(t)=-1$ can be thought of as the resting state of the quantum memory, when it is not loaded.
	
	The agent's \textit{policy}, i.e., the policy of the elementary link corresponding to the edge $e\in E$, is a sequence of the form $\pi_{e}=(d_1,d_2,\dotsc)$, where the \textit{decision functions} $d_t$ are defined as
	\begin{equation}
		d_t(h^t)(a_t)\coloneqq\Pr[A_{e}(t)=a_t|H_{e}(t)=h^t].
	\end{equation}
	In other words, the decision functions give us the probability that the agent takes a particular action, given the history of actions and statuses.
	
	
	Let $\{\pi_{e}:e\in E\}$ denote a collection of policies for all of the elementary links of the network. Then, the quantum state of the network at time steps $t\geq 1$ is 
	\begin{equation}\label{eq-network_cq_state_QDP}
		\bigotimes_{e\in E}\widehat{\sigma}_{\!e}^{\pi_{e}}(t),
	\end{equation}
	where $\widehat{\sigma}_{\!e}^{\pi_{e}}(t)$ is a classical-quantum state for the elementary link corresponding to $e\in E$, which has the form
	\begin{equation}\label{eq-elem_link_cq_state_QDP}
		\widehat{\sigma}_{\!e}^{\pi_{e}}(t)=\sum_{h^t\in\{0,1\}^{2t-1}}\Pr[H_{e}(t)=h^t]_{\pi_{e}}\ket{h^t}\bra{h^t}_{H_t}\otimes\sigma_{\!e}(t|h^t),
	\end{equation}
	where for every history $h^t=(x_1,a_1,x_2,a_2,\allowbreak\dotsc,a_{t-1},x_t)\in\{0,1\}^{2t-1}$,
	\begin{equation}
		\ket{h^t}_{H_t}\equiv\ket{x_1}_{X_1}\ket{a_1}_{A_1}\ket{x_2}_{X_2}\ket{a_2}_{A_2}\dotsb\ket{a_{t-1}}_{A_{t-1}}\ket{x_t}_{X_t}.
	\end{equation}
	This classical-quantum state captures both the history of the agent-environment interaction corresponding to the elementary link (in the classical register $H_t$), as well as the quantum state of the nodes belonging to the elementary link conditioned on a particular history. We present explicit expressions for $\Pr[H_e(t)=h^t]_{\pi_e}$ and $\sigma_{\!e}(t|h^t)$ in Section~\ref{sec-q_state_elem_link} below.
	
	The tensor product structure in \eqref{eq-network_cq_state_QDP} holds because, by assumption, all of the agents corresponding to the elementary links are independent. This means that all of the agents have knowledge only of the status of their own elementary link. We stick to this setting throughout this work. However, we can use quantum decision processes to develop quantum network protocols in which the agents can cooperate, so that they have knowledge of the network in a certain local neighborhood of their elementary link. The quantum state of the network would then not have the simple tensor product structure as in \eqref{eq-network_cq_state_QDP}. Furthermore, in the case of independent agents, we can use the quantum decision processes for the elementary links as building blocks for quantum decision processes for groups of elementary links, leading to more sophisticated quantum network protocols. We leave these investigations as interesting directions for future work.

\subsection{Quantum state of an elementary link}\label{sec-q_state_elem_link}

	We now state one of the main results of this work, which is an explicit expression for the probabilities $\Pr[H_{e}(t)=h^t]_{\pi}$ and the quantum states $\sigma_{\!e}(t|h^t)$ of the elementary link corresponding to the edge $e$ in a quantum network undergoing a policy $\pi=(d_1,d_2,\dotsc,d_t,\dotsc)$. This result gives us the answer to the first question posed at the end of Section~\ref{sec-ent_dist_task}.
	
	\begin{theorem}[Quantum state of an elementary link]\label{thm-link_quantum_state}
		Let $G=(V,E)$ be the graph corresponding to the elementary links of a quantum network, and let $e\in E$. For all $t\geq 1$, histories $h^t=(x_1,a_1,\dotsc,a_{t-1},x_t)\in\{0,1\}^{2t-1}$, and policies $\pi=(d_1,d_2,\dotsc,d_t,\dotsc)$, we have
		\begin{equation}\label{eq-link_state_unnormalized}
			\sigma_{\!e}(t|h^t)=x_t\,\rho_{e}\!\left(M_{e}(t)(h^t)\right)+(1-x_t)\tau_{e}^{\varnothing}.
		\end{equation}
		Furthermore,
		\begin{equation}\label{eq-hist_prob_general}
			\Pr[H_{e}(t)=h^t]_{\pi}=\left(\prod_{j=1}^{t-1} d_j(h_j^t)(a_j)\right) p_{e}^{N_{e}^{\text{succ}}(t)(h^t)}(1-p_{e})^{N_{e}^{\text{req}}(t)(h^t)-N_{e}^{\text{succ}}(t)(h^t)}
		\end{equation}
		for every history $h^t=(x_1,a_1,\dotsc,a_{t-1},x_t)\in\{0,1\}^{2t-1}$, where $p_e$ is the success probability, $h_j^t=(x_1,a_1,\dotsc,a_{j-1},x_j)$, and
		\begin{equation}
			N_{e}^{\text{req}}(t)\coloneqq\sum_{j=1}^t A_{e}(j-1),\quad N_{e}^{\text{succ}}(t)\coloneqq \sum_{j=1}^t A_{e}(j-1)X_{e}(j)
		\end{equation}
		are the number of elementary link requests and the number of successful elementary link requests, respectively, up to time~$t$, with $A_e(0)\equiv 1$.~\defqed
	\end{theorem}
	
	
	The expected quantum state of the elementary link corresponding to an edge $e\in E$ undergoing the policy $\pi$ is defined as the state obtained by tracing out the classical history register in \eqref{eq-elem_link_cq_state_QDP}:
	\begin{equation}\label{eq-link_avg_q_state_def}
		\sigma_{\!e}^{\pi}(t)\coloneqq\sum_{h^t\in\{0,1\}^{2t-1}} \Pr[H_e(t)=h^t]_{\pi}\,\sigma_{\!e}(t|h^t).
	\end{equation}
	Using Theorem~\ref{thm-link_quantum_state}, we immediately obtain the following result.
	
	\begin{corollary}[Expected quantum state of an elementary link]\label{cor-link_avg_q_state}
		Let $G=(V,E)$ be the graph corresponding to the elementary links of a quantum network, and let $e\in E$. For all $t\geq 1$ and for every policy $\pi$, the expected quantum state of the elementary link corresponding to $e$ is
		\begin{equation}\label{eq-link_avg_q_state}
			\sigma_{\!e}^{\pi}(t)=(1-\Pr[X_{e}(t)=1]_{\pi})\tau_{e}^{\varnothing}+\sum_{m}\Pr[X_{e}(t)=1,M_{e}(t)=m]_{\pi}\,\rho_{e}(m),
		\end{equation}
		where the sum is with respect to all possible values of the memory time, which in general depends on the policy $\pi$.~\defqed
	\end{corollary}
	
	

\subsection{Figures of merit}\label{sec-figures_of_merit_general}

	
	
	Consider an edge $e\in E$ in the graph $G=(V,E)$ corresponding to the elementary links of a quantum network, and let $\pi$ be a policy for the elementary link corresponding to $e$. Having determined the quantum state of the elementary link corresponding to $e$, let us now consider the following figures of merit to evaluate the policy $\pi$.
	\begin{itemize}
		\item The probability that an elementary link is active at time $t\geq 1$, i.e., $\Pr[X_e(t)=1]_{\pi}=\mathbb{E}[X_e(t)]_{\pi}$. Due to the latter equality, we also refer to his quantity as the expected elementary link status.
		
		\item The expected fidelity of the quantum state of the elementary link with respect to a target quantum state at time $t\geq 1$, i.e., $\mathbb{E}[\widetilde{F}_{\!e}(t)]_{\pi}$, where
			\begin{equation}\label{eq-elem_link_fid_RV}
				\widetilde{F}_{\!e}(t)\coloneqq X_e(t)f_e(M_e(t)),
			\end{equation}
			where 
			\begin{equation}\label{eq-fidelity_decay}
				f_{e}(m)\coloneqq \bra{\psi}\rho_{e}(m)\ket{\psi}=\bra{\psi}\mathcal{N}_{e}^{\circ m}(\rho_{e}^0)\ket{\psi}
			\end{equation}
			denotes the fidelity of the state $\rho_e(m)$ with respect to a pure target state $\psi=\ket{\psi}\bra{\psi}$.
			
			A related quantity is
			\begin{equation}\label{eq-elem_link_fid_conditional}
				\mathbb{E}[F_{\!e}(t)]_{\pi}\coloneqq\frac{\mathbb{E}[\widetilde{F}_{\!e}(t)]_{\pi}}{\Pr[X_e(t)=1]_{\pi}},
			\end{equation}
			which can be thought of as the expected fidelity of the quantum state of the elementary link \textit{given} that the elementary link is active.
	\end{itemize}
	We discuss other figures of merit of interest in Appendix~\ref{sec-other_figs_of_merit}.
	
	The expected status and the expected fidelity of an elementary link can both be expressed in a simple manner in terms of the classical-quantum state of the elementary link.

	\begin{theorem}\label{thm-succ_prob_fid_cq}
		Let $G=(V,E)$ be the graph corresponding to the elementary links of a quantum network, let $e\in E$, and let $t\geq 1$. Given a policy $\pi$, as well as a target pure state $\psi$, for the elementary link corresponding to $e$, we have
		\begin{align}
			\mathbb{E}[X_e(t)]_{\pi}&=\Tr[\ket{1}\bra{1}_{X_t}\widehat{\sigma}_{\!e}^{\pi}(t)]=\Pr[X_e(t)=1]_{\pi},\label{eq-link_value_prob_via_state}\\
			\mathbb{E}[\widetilde{F}_{\!e}(t)]_{\pi}&=\Tr[(\ket{1}\bra{1}_{X_t}\otimes\psi)\widehat{\sigma}_{\!e}^{\pi}(t)].\label{eq-avg_Ftilde} \defqedspec
		\end{align}
	\end{theorem}
	

	Theorem~\ref{thm-succ_prob_fid_cq} answers the second question posed at the end of Section~\ref{sec-ent_dist_task}.

\subsection{Examples of policies}

	Let us examine the figures of merit defined above using two simple policies.

	First, consider the policy consisting of the action ``request'' at every time step before the elementary link becomes active, and the action ``wait'' at every time step after the elementary link becomes active. This policy is defined by $\pi=(d_1,d_2,\dotsc)$, where
	\begin{equation}
		d_t(h^t)=\left\{\begin{array}{l l} 1 & \text{if }X_e(t)(h^t)=0, \\ 0 & \text{if }X_e(t)(h^t)=1, \end{array}\right.
	\end{equation}
	for all $t\geq 1$ and every history $h^t\in\{0,1\}^{2t-1}$. This policy achieves the highest value of $\mathbb{E}[X_e(t)]_{\pi}$ for all $t\geq 1$. In fact, this policy is simply the $t^{\star}=\infty$ memory-cutoff policy, which we investigate in Section~\ref{sec-mem_cutoff_policy}. We show in that section that $\mathbb{E}[X_e(t)]_{\pi}=1-(1-p_e)^t$ for all $t\geq 1$. Of course, this highest value of $\mathbb{E}[X_e(t)]_{\pi}$ comes at the cost of a lower fidelity, because each ``wait'' action decreases the fidelity of the quantum state stored in memory; we see this explicitly in Section~\ref{sec-mem_cutoff_fidelity}.
	
	Another policy is one in which the action ``request'' is taken at every time step, i.e., $d_t(h^t)=1$ for all $t\geq 1$ and every history $h^t\in\{0,1\}^{2t-1}$. (This is the $t^{\star}=0$ memory-cutoff policy; see Section~\ref{sec-mem_cutoff_policy}.) In this case, the quantity $\mathbb{E}[F_{\!e}(t)]_{\pi}$ is maximized, because $\mathbb{E}[F_{\!e}(t)]_{\pi}=f_e(0)$ for every time step $t\geq 1$, which is the highest that can be obtained (without entanglement distillation). This highest value of the fidelity comes at the cost of a lower success probability, because the probability that the elementary link is active stays at $p_e$ for all times with respect to this policy, i.e., $\Pr[X_e(t)=1]_{\pi}=p_e$ for all $t\geq 1$ if at every time step the agent requests a link.
	
	The two policies considered above illustrate the trade-off between the expected status $\mathbb{E}[X_e(t)]_{\pi}$ and the expected fidelity $\mathbb{E}[F_{\!e}(t)]_{\pi}$ of an elementary link. The quantity $\mathbb{E}[\widetilde{F}_{\!e}(t)]_{\pi}$, with $\widetilde{F}_{\!e}(t)$ defined in \eqref{eq-elem_link_fid_RV}, incorporates this trade-off, as it can be thought of intuitively as the product of the status and fidelity of an elementary link. Let us therefore now turn to finding policies that maximize the quantity $\mathbb{E}[\widetilde{F}_{\!e}(t)]_{\pi}$ as a function of time. 

\subsection{Policy optimization}\label{sec-policy_opt}

	We would like to understand the highest value that $\mathbb{E}[\widetilde{F}_{\!e}(t)]_{\pi}$ can take as a function of policies $\pi$ and times $t\geq 1$, given a particular entanglement generation success probability $p_e$ (and thus, given particular values of the physical parameters that comprise the success probability). With this, we can answer the third question posed at the end of Section~\ref{sec-ent_dist_task}, and more broadly, we can begin to understand the limits of practical, near-term quantum networks.

\subsubsection{Backward recursion}

	Optimization of $\mathbb{E}[\widetilde{F}_{\!e}(t)]_{\pi}$ with respect to policies $\pi$ for a given elementary link is given by the following recursive procedure.
	
	\begin{theorem}[Optimal policy for an elementary link]\label{thm-opt_policy}
		Let $G=(V,E)$ be the graph corresponding to the elementary links of a quantum network, let $e\in E$, and let $\psi\equiv\ket{\psi}\bra{\psi}$ be a pure target state. Then, for all $T\geq 1$,
		\begin{equation}
			\max_{\pi}\mathbb{E}[\widetilde{F}_{\!e}(T+1)]_{\pi}=\sum_{x_1=0}^1\max_{a_1\in\{0,1\}}w_{2}(x_1,a_1),
		\end{equation}
		where
		\begin{align}
			w_{t}(h^{t-1},a_{t-1})&=\sum_{x_{t}=0}^1\max_{a_{t}\in\{0,1\}}w_{t+1}(h^{t-1},a_{t-1},x_t,a_t) \quad\forall~2\leq t\leq T,\label{eq-pol_opt_intermediate}\\[0.2cm]
			w_{T+1}(h^{T},a_{T})&=p_{e}^{y}(1-p_{e})^{x-y} f_e\!\left(M_e(T+1)(h^T,a_T,1)\right),\label{eq-pol_opt_final}
		\end{align}
		and $y=N_{e}^{\text{succ}}(T+1)(h^T,a_T,1)$, $x=N_{e}^{\text{req}}(T+1)(h^T,a_T,1)$. Furthermore, the optimal policy is deterministic and given by
		\begin{equation}\label{eq-backward_recursion_policy}
			d_t^{\text{BR}}(h^t)\coloneqq\argmax_{a_t\in\{0,1\}}w_{t+1}(h^t,a_t)\quad \forall~1\leq t\leq T. \defqedspec
		\end{equation}
	\end{theorem}
	
	
	Observe that in the recursive procedure presented in Theorem~\ref{thm-opt_policy}, we must first determine the optimal action at the final time step and then determine the optimal actions at the previous time steps in turn. For this reason, the recursive procedure is often known as \textit{backward recursion}. Note that this recursive procedure is exponentially slow in the final time because of the fact that the number of histories grows exponentially with time---the number of histories up to time $t$ is $|\{0,1\}^{2t-1}|=2^{2t-1}$. For this reason, it is useful to have efficient methods for estimating the maximum fidelity of an elementary link. One such method is forward recursion. 

\subsubsection{Forward recursion}

	Instead of starting from the final time step and finding the optimal actions by going backwards, we could instead find optimal actions by going forwards, i.e., by selecting the action such that the immediate expected fidelity is maximized. Such a ``forward recursion'' approach is more natural from the perspective of a real-world learning agent, who has to make decisions in real time and does not necessarily have complete knowledge of the environment in order to perform the backward recursion algorithm. However, the forward recursion algorithm will not necessarily lead to a globally optimal policy. In fact, the globally optimal policy can be obtained using the backward recursion algorithm, which is the result of Theorem~\ref{thm-opt_policy}. Nevertheless, it is worthwhile to briefly discuss the forward recursion algorithm because many reinforcement learning algorithms are based on it, and they give efficiently computable lower bounds on the maximum expected fidelity of an elementary link.

	\begin{theorem}[Policy optimization via forward recursion]\label{thm-network_QDP_forward_recursion_policy}
		Let $G=(V,E)$ be the graph corresponding to the elementary links of a quantum network, and let $e\in E$. The forward recursion algorithm gives rise to the deterministic policy $(d_1^{\text{FR}},d_2^{\text{FR}},\dotsc)$ for the elementary link corresponding to $e$, where
		\begin{equation}\label{eq-network_QDP_forward_recursion_policy}
			d_t^{\text{FR}}(h^t)\coloneqq\left\{\begin{array}{l l} 1 & \text{if }x_t=0,\\[0.2cm] 0 & \text{if }x_t=1\text{ and }f_e\!\left(M_{e}(t)(h^t)+1\right)>p_{e} f_e(0),\\[0.2cm] 1 & \text{if } x_t=1\text{ and }f_e\!\left(M_{e}(t)(h^t)+1\right)\leq p_{e} f_e(0), \end{array}\right.
		\end{equation}
		for all $t\geq 1$ and all $h^t\in\{0,1\}^{2t-1}$, where $p_e$ is the success probability for the elementary link corresponding to $e$.~\defqed
	\end{theorem}
	

\section{The memory-cutoff policy}\label{sec-mem_cutoff_policy}
	
	In Section~\ref{sec-elem_link_QDP}, we defined a quantum decision process for elementary links in a quantum network. We determined results for the quantum state of an elementary link for arbitrary policies, defined figures of merit to evaluate policies (in particular, the expected fidelity), and presented algorithms for determining optimal policies. Let us now consider an explicit example of a policy.
	
	A natural policy to consider, and one that has been considered extensively previously \cite{CJKK07,SRM+07,SSM+07,BPv11,JKR+16,DHR17,RYG+18,SSv19,KMSD19,SJM19,LCE20}, is the following deterministic policy. An elementary link is requested at every time step until it becomes active, and once it is active it is held in quantum memory for some pre-specified amount $t^{\star}$ of time steps (usually called the \textit{memory cutoff} and not necessarily equal to the memory coherence time) after which the quantum state of the elementary link is discarded and requested again. The cutoff $t^{\star}$ can be any value in the set $\mathbbm{N}_0\cup\{\infty\}$, where $\mathbbm{N}_0=\{0,1,2,\dotsc\}$. There are two extreme cases of this policy: when $t^{\star}=0$, a request is made at \textit{every} time step regardless of whether the previous request succeeded; if $t^{\star}=\infty$, then an elementary link request is made at every time step until the elementary link becomes active, and once it becomes active the corresponding entangled state remains in memory indefinitely---no further request is ever made. In this section, we provide a complete analysis of this policy for all values of $t^{\star}\in\mathbbm{N}_0\cup\{\infty\}$ using the general developments in Section~\ref{sec-elem_link_QDP}. Details of the analysis can be found in Appendix~\ref{app-mem_cutoff_details}.
	
	Throughout this section, we consider an arbitrary elementary link in a quantum network, specified by the edge $e\in E$ in the graph $G=(V,E)$ corresponding to the elementary links of the network.
	
	\smallskip
	
	\begin{definition}[Memory-cutoff policy]\label{def-mem_cutoff_policy}
		Given a cutoff $t^{\star}\in\mathbbm{N}_0\cup\{\infty\}$, the \textit{$t^{\star}$ memory-cutoff policy} is the sequence $(d_1^{t^{\star}},d_2^{t^{\star}},\dots)$ given by the following for all $t\geq 1$ and every history $h^t\in\{0,1\}^{2t-1}$.
		\begin{align}
			&\bullet\quad\text{If }t^{\star}=0\,\,:\,\,d_t^0(h^t)\coloneqq 1.\\
			&\bullet\quad\text{If }t^{\star}\in\mathbb{N}\,\,:\,\,d_t^{t^{\star}}(h^t)\coloneqq\left\{\begin{array}{l l} 0 & \text{if } 0\leq M_e(t)(h^t)\leq t^{\star}-1, \\ 1 & \text{if }M_e(t)(h^t)=-1\text{ or }M_e(t)(h^t)=t^{\star}. \end{array}\right.\label{eq-mem_cutoff_policy_finite}\\[0.3cm]
			&\bullet\quad\text{If }t^{\star}=\infty\,\,:\,\,d_t^{\infty}(h^t)\coloneqq\left\{\begin{array}{l l} 0 & \text{if } M_e(t)(h^t)\geq 0, \\ 1 & \text{if }M_e(t)(h^t)=-1. \defqedspec \end{array}\right. \label{eq-mem_cutoff_policy_infty}
		\end{align}
	\end{definition}

\subsection{Expected quantum state}
	
	Recall from \eqref{eq-link_avg_q_state} that the expected quantum state of an elementary link in a quantum network undergoing a policy $\pi$ is given by the probabilities $\Pr[X_e(t)=1,M_e(t)=m]_{\pi}$ and $\Pr[X_e(t)=1]_{\pi}$. We now provide analytic expressions for these probabilities, which we denote by $\Pr[X_e(t)=1,M_e(t)=m]_{t^{\star}}$ and $\Pr[X_e(t)=1]_{t^{\star}}$, in the case of the memory-cutoff policy for all possible values of the cutoff $t^{\star}$. We consider both the short-term ($t<\infty$) and the long-term ($t\to\infty$) behavior of these probabilities.

\subsubsection{Short-term behavior}

	In the short term, we obtain the following result for the joint probability distribution of the elementary link status $X_e(t)$ and memory time $M_e(t)$ random variables.
		
	\begin{theorem}\label{thm-mem_status_pr}
		Let $p_e\in[0,1]$ be the success probability for an elementary link in a quantum network, as defined in Section~\ref{sec-q_state_elem_link}, and let $t\geq 1$.
		\begin{itemize}
			\item For all $t^{\star}\in\mathbb{N}_0\cup\{\infty\}$, if $t\leq t^{\star}+1$ and $m\in\{-1,0,1,\dotsc,t-1\}$, then
				\begin{equation}\label{eq-mem_status_pr_a}
					\Pr[M_e(t)=m,X_e(t)=1]_{t^{\star}}=p_e(1-p_e)^{t-(m+1)}(1-\delta_{m,-1}).
				\end{equation}
				If $t>t^{\star}+1$ and $m\in\{-1,0,1,\dotsc,t^{\star}\}$, then
				\begin{multline}\label{eq-mem_status_pr}
					\Pr[M_e(t)=m,X_e(t)=1]_{t^{\star}}\\=(1-\delta_{m,-1})\sum_{x=0}^{\floor{\frac{t-1}{t^{\star}+1}}} \binom{t-(m+1)-xt^{\star}}{x}\boldsymbol{1}_{t-(m+1)-x(t^{\star}+1)\geq 0}\\\times p_e^{x+1} (1-p_e)^{t-(m+1)-x(t^{\star}+1)},
				\end{multline}
				where
				\begin{equation}
					\boldsymbol{1}_{t-(m+1)-x(t^{\star}+1)\geq 0}=\left\{\begin{array}{l l} 1 & \text{if }t-(m+1)-x(t^{\star}+1)\geq 0, \\ 0 & \text{otherwise.}\end{array}\right.
				\end{equation}
				
			\item For all $t^{\star}\in\mathbb{N}_0$ and $m\in\{-1,0,1,\dotsc,t^{\star}\}$,
				\begin{multline}\label{eq-mem_time_prob_x0}
					\Pr[M_e(t)=m,X_e(t)=0]_{t^{\star}}\\=\left\{\begin{array}{l l} \delta_{m,-1}(1-p_e)^t,\quad t\leq t^{\star}+1,\\[0.5cm] \displaystyle \delta_{m,-1}\sum_{x=0}^{\floor{\frac{t-1}{t^{\star}+1}}}\binom{t-1-xt^{\star}}{x}p_e^x(1-p_e)^{t-(t^{\star}+1)x},\quad t>t^{\star}+1. \end{array}\right.
				\end{multline}
				For $t^{\star}=\infty$ and $m\in\{-1,0,1,\dotsc,t-1\}$,
				\begin{equation}
					\Pr[M_e(t)=m,X_e(t)=0]_{\infty}=\delta_{m,-1}(1-p_e)^t. \defqedspec
				\end{equation}
		\end{itemize}
	\end{theorem}
	
	
	From Theorem~\ref{thm-mem_status_pr}, we immediately obtain an expression for the probability that an elementary link is active at all times $t\geq 1$.
	
\newpage
	
	\begin{corollary}\label{cor-link_status_Pr1}
		For every time $t\geq 1$, cutoff $t^{\star}\in\mathbbm{N}_0\cup\{\infty\}$, and success probability $p_e\in[0,1]$, the probability that an elementary link in a quantum network undergoing the $t^{\star}$ memory-cutoff policy is active at time $t$ is
		\begin{multline}\label{eq-link_status_Pr1}
			\Pr[X_e(t)=1]_{t^{\star}}\\=\left\{\begin{array}{l l} 1-(1-p_e)^t, & t\leq t^{\star}+1,\\[0.5cm] \displaystyle \sum_{x=0}^{\floor{\frac{t-1}{t^{\star}+1}}}\sum_{k=1}^{t^{\star}+1}\binom{t-k-xt^{\star}}{x}\boldsymbol{1}_{t-k-x(t^{\star}+1)\geq 0}p_e^{x+1}(1-p_e)^{t-k-(t^{\star}+1)x}, & t>t^{\star}+1. \defqedspec \end{array}\right.
		\end{multline}
	\end{corollary}
	
	
	\begin{figure}
		\centering
		\includegraphics[scale=1]{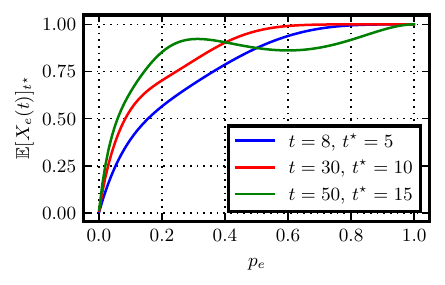}\quad
		\includegraphics[scale=1]{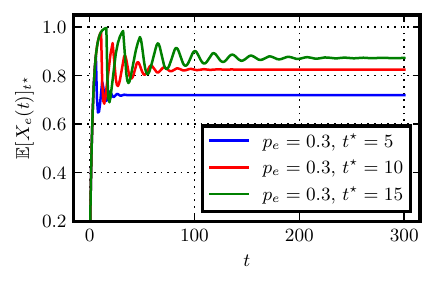}
		\caption{(Left) The expected elementary link status, given by \eqref{eq-link_status_Pr1}, as a function of the  success probability $p_e$ for various values of $t$ and $t^{\star}$. (Right) The expected elementary link status, given by \eqref{eq-link_status_Pr1}, as a function of $t$ for various values of $p_e$ and $t^{\star}$.}\label{fig-avg_link_value_example}
	\end{figure}
	
	See Figure~\ref{fig-avg_link_value_example} for plots of the expected elementary link status $\mathbb{E}[X_e(t)]_{t^{\star}}=\Pr[X_e(t)=1]_{t^{\star}}$ as a function of time $t$ and as a function of the  success probability $p_e$. The oscillatory behavior of the expected elementary link status as a function of time in the right panel of Figure~\ref{fig-avg_link_value_example} is due to the nature of the memory-cutoff policy, which requires that the elementary link be discarded every $t^{\star}$ time steps. Indeed, the period of the oscillations is $t^{\star}$, which is apparent for short times and large values of $t^{\star}$. In the long term, however, we see that the amplitude of the oscillations decreases, and the expected elementary link status reaches a steady state, whose value we state below in Theorem~\ref{thm-avg_mem_status_inf}.

\subsubsection{Long-term (steady-state) behavior}

	Let us now consider the $t\to\infty$, or long-term behavior of an elementary link undergoing the memory-cutoff policy.
	
	\begin{theorem}\label{thm-avg_mem_status_inf}
		Let $t^{\star}\in\mathbbm{N}_0$ be the cutoff and $p_e\in[0,1]$ be the  success probability for an elementary link in a quantum network. For all $m\in\{-1,0,1,\dotsc,t^{\star}\}$,
		\begin{align}
			\lim_{t\to\infty}\Pr[M_e(t)=m,X_e(t)=1]_{t^{\star}}&=\frac{p_e}{1+t^{\star}p_e}(1-\delta_{m,-1}),\label{eq-avg_mem_status_inf}\\
			\lim_{t\to\infty}\Pr[M_e(t)=m,X_e(t)=0]_{t^{\star}}&=\frac{1-p_e}{1+t^{\star}p_e}\delta_{m,-1}. \label{eq-mem_time_prob_infty_x0}
		\end{align}
		Consequently,
		\begin{equation}\label{eq-link_status_avg_inf}
			\lim_{t\to\infty}\mathbb{E}[X_e(t)]_{t^{\star}}=\frac{(t^{\star}+1)p_e}{1+t^{\star}p_e}. \defqedspec
		\end{equation}
	\end{theorem}

\subsection{Expected fidelity}\label{sec-mem_cutoff_fidelity}
	
	If an elementary link in a quantum network undergoes the $t^{\star}$ memory-cutoff policy, and it has  success probability $p_e$, then from Theorem~\ref{thm-mem_status_pr} and Corollary~\ref{cor-link_status_Pr1}, we immediately obtain the following expressions for the expected fidelity of the elementary link for all $t\geq 1$:
	\begin{align}
		\mathbb{E}[\widetilde{F}_{\!e}(t)]_{t^{\star}}&=\left\{\begin{array}{l l} \displaystyle \sum_{m=0}^{t-1}f_e(m)p_e(1-p_e)^{t-(m+1)} & t\leq t^{\star}+1,\\[0.5cm] \displaystyle \sum_{m=0}^{t^{\star}} f_e(m) \Pr[M_e(t)=m,X_e(t)=1]_{t^{\star}} & t>t^{\star}+1, \end{array}\right.\label{eq-avg_fid_tilde}\\[0.5cm]
		\mathbb{E}[F_{\!e}(t)]_{t^{\star}}&=\left\{\begin{array}{l l}\displaystyle \sum_{m=0}^{t-1}f_e(m)\frac{p_e(1-p_e)^{t-(m+1)}}{1-(1-p_e)^t} & t\leq t^{\star}+1,\\[0.5cm] \displaystyle \sum_{m=0}^{t^{\star}} f_e(m)\frac{\Pr[M_e(t)=m,X_e(t)=1]_{t^{\star}}}{\Pr[X_e(t)=1]_{t^{\star}}} & t>t^{\star}+1, \end{array}\right.\label{eq-avg_fid}
	\end{align}
	where in \eqref{eq-avg_fid_tilde} and \eqref{eq-avg_fid} the expression for $\Pr[M_e(t)=m,X_e(t)=1]_{t^{\star}}$ for $t>t^{\star}+1$ is given in \eqref{eq-mem_status_pr}, and the expression for $\Pr[X_e(t)=1]_{t^{\star}}$ for $t>t^{\star}+1$ is given in \eqref{eq-link_status_Pr1}.
	
	In the limit $t\to\infty$, using Theorem~\ref{thm-avg_mem_status_inf}, we obtain the following:
	\begin{align}
		\lim_{t\to\infty}\mathbb{E}[\widetilde{F}_{\!e}(t)]_{t^{\star}}&=\frac{p_e}{1+t^{\star}p_e}\sum_{m=0}^{t^{\star}}f_e(m),\quad t^{\star}\in\mathbbm{N}_0,\label{eq-avg_fid_tilde_tInfty}\\
		\lim_{t\to\infty}\mathbb{E}[F_{\!e}(t)]_{t^{\star}}&=\frac{1}{t^{\star}+1}\sum_{m=0}^{t^{\star}}f_e(m),\quad t^{\star}\in\mathbbm{N}_0.\label{eq-avg_fid_tInfty}
	\end{align}

\section{Summary and outlook}

	Understanding the capabilities and limitations of near-term quantum networks is an important problem, whose solutions will help drive the physical realization of small-scale quantum networks, and eventually lead to the realization of a global-scale quantum internet. Before such developments can be made, we must first have a common language and theoretical framework for analyzing quantum network protocols. This work provides the first steps towards such a general theoretical framework for practical quantum network protocols. We make use of the concept of a decision process to model protocols for elementary links. This formulation is natural based on the structure of near-term quantum network protocols, and it allows for optimal policies to be obtained using dynamic programming algorithms.
	
	Section~\ref{sec-elem_link_QDP} constitutes the main conceptual and technical contributions of this work. It lays out a quantum decision process for elementary links in a quantum network. The framework allows us to model protocols for an elementary link with respect to time in terms of the actions of an agent that can either request an entangled state from a source or keep the one it currently has in its quantum memory. The sequence of actions of the agent over time defines a policy, which is synonymous with the protocol. After formulating the quantum decision process for an elementary link and defining figures of merit for evaluating policies, we considered three examples of policies: the backward recursion policy in \eqref{eq-backward_recursion_policy}, the forward recursion policy in \eqref{eq-network_QDP_forward_recursion_policy}, and the memory-cutoff policy in Definition~\ref{def-mem_cutoff_policy}. We proved that the backward recursion policy is optimal among all policies (Theorem~\ref{thm-opt_policy}).
	
	In Section~\ref{sec-mem_cutoff_policy}, we considered the memory-cutoff policy, and we applied the results of Section~\ref{sec-elem_link_QDP} in order to obtain closed-form expressions for the expected quantum state of an elementary link for an arbitrary value of the cutoff, for both short times and long times. We also obtained closed-form expressions for the figures of merit defined in Section~\ref{sec-figures_of_merit_general}.
	
	We expect the results of this work to be useful as a building block for large-scale quantum network protocols. For example, the policies for elementary links considered in this work can be used as an underlying policy layer on top of which routing protocols can be applied in order to obtain an overall (in general non-optimal) policy for generating end-to-end entanglement in a network. Furthermore, because our results apply to elementary links consisting of an arbitrary number of nodes and to any noise model for the quantum memories, they can be applied to protocols that go beyond bipartite entanglement distribution, namely to protocols for distributing multipartite entanglement. We also expect our results to be useful in the analysis of entanglement distribution using all-photonic quantum repeaters \cite{KKL15}, and in the analysis of entanglement distribution using satellite-based quantum networks \cite{BBM+15,KBD+19,LKB20,GSH+20,PMD+20}, in which an elementary link can easily be on the order of 1000 km \cite{YCL17} while still having a high fidelity. Initial applications of the results of this work to satellite-based elementary links can be found in \cite[Chapter~7]{Kha21}.
	
	This work also opens up several other interesting directions for future work. Of immediate interest is to go beyond the elementary link level by incorporating entanglement distillation and swapping into the decision process developed here, which would allow for the analysis of more sophisticated quantum network protocols, and it would build on prior works \cite{AME11,MT13,PWD18,PD19,WMDB19} that examine quantum network protocols beyond the elementary link level. Such an extension would involve multiple cooperating agents, in contrast to the independent agents considered in this work, and can in principle be formulated for an arbitrary network topology. A simple, but relevant example of a network topology, which has also been considered recently, is the star-shaped network used for the so-called ``quantum entanglement switch'' \cite{VGNT20,PGGT20,VGNT21,VGPT21b}. As we might expect, these extra elements of entanglement distillation and swapping will make analytic analysis (as done in this work) intractable. This is when reinforcement learning algorithms are expected to be helpful for finding optimal policies. The beginnings of some of these future developments can be found in \cite[Appendix~D]{Kha21}.

\section*{Acknowledgments}

	I dedicate this work to the memory of Jonathan P. Dowling. This work, and the other works that I was fortunate to co-author with Jon, would not have been possible without his constant encouragement and his enthusiasm for the quantum internet.

	\noindent The plots in this work were made using the Python package matplotlib \cite{matplotlib}.
	
	\noindent Financial support was provided by the National Science Foundation and the National Science and Engineering Research Council of Canada Postgraduate Scholarship.

\newpage

\begin{appendices}
\addtocontents{toc}{\protect\setcounter{tocdepth}{1}}

\section{Related work}\label{app-related_work}

	Prior theoretical work on quantum networks can essentially be split into two types. The first type of work is information theoretic \cite{AML16,AK17,BA17,RKB+18,Pir19,Pir19b,DBWH19}, with the focus being on obtaining (or placing bounds on) the ultimate limits of communication in a quantum network, without taking device imperfections (such as quantum memories with limited coherence times and non-deterministic gate operations) explicitly into account. Consequently, this type of work does not always provide a realistic analysis for near-term quantum networks. The second type of theoretical work on quantum networks \cite{BDC98,DBC99,DLCZ01,GKL+03,RHG05,CJKK07,SRM+07,SSM+07,JTN+09,FWH+10,BPv11,ZDB12,MSD+12,MKL+14,KKL15,ZBD16,EKB16,NJKL16,JKR+16,MLK+16,MEL17,LZH+17,VK17,DHR17,RGR+18,RYG+18,SSv19,HBE20} (see also \cite{SSR+11,MATN15,VanMeter_book} and the references therein) focuses on calculating communication rates under more realistic assumptions on the devices, sometimes with specific physical architectures and different types of noise- and loss-mitigation techniques, such as entanglement distillation and quantum error-correction, taken into account. Typically, these works have been focused primarily on the topology of a linear chain of nodes. However, more recent work \cite{WZM+16,ZPD+18,PWD18,DKD18,WPZD19,PD19,KMSD19,DSM+19,MMG19,CERW20,GEW20,STCMW20,DPW20} has begun to focus on arbitrary topologies, with routing protocols taken into account in some cases \cite{SMI+17,PKT+19,HPE19,CRDW19,Pir19,Pir19b,LLLC20,LBD+20}. The techniques used in these works are often varied, and sometimes different terminology and mathematical tools are used. One of the aims of this work is to provide the starting point for a unified theoretical framework for practical quantum network protocols that can, in principle, incorporate and generalize the developments in the aforementioned works of the second type, and that can be applied to arbitrary topologies and physical architectures. 

	Policy-based approaches to quantum network protocols, as considered in this work, have been considered before in \cite{VLMN09,AME11,JKR+16,MDV19} (see also \cite{VanMeter_book}), where terms such as ``rule-set'' or ``schedule'' have been used instead of ``policy''. In \cite{JKR+16}, the authors consider different control protocols for elementary links in a quantum network based on different configurations of the sources and heralding stations and the impact they have on end-to-end entanglement distribution rates. In \cite{VLMN09}, the authors look at protocols for end-to-end entanglement distribution along a chain of quantum repeaters and simulate different scheduling protocols for entanglement distillation along elementary links. Similarly, in \cite{AME11}, the authors use finite state machines to analyze the different layers of an end-to-end entanglement distribution protocol in quantum networks, such as entanglement distillation and entanglement swapping. Finally, in \cite{MDV19}, the authors use an approach based on rule-sets to determine end-to-end entanglement distribution rates and fidelities of the end-to-end pairs along a chain of quantum repeaters. One of the goals of this work is to explicitly formalize the approaches taken in the aforementioned works within the context of decision processes, because this allows us to systematically study different policies and calculate quantities that are relevant for quantum networks, such as entanglement distribution rates and fidelities of the quantum states of the links.
	
	This work is complementary to prior work that uses Markov chains to analyze waiting times and entanglement distribution rates for a chain of quantum repeaters \cite{RFT+09,VK19,BCE19,SSv19}; we also refer to the work on entanglement switches in \cite{VGNT20,PGGT20,VGNT21,VGPT21b}, which use both discrete-time and continuous-time Markov chains. This work is also complementary to prior work that analyzes the quantum state in a quantum repeater chain with noisy quantum memories~\cite{HKBD07,RPL09,GKF+15,KGD+16,KVS+19}. 
	
	In \cite{WMDB19}, the authors use reinforcement learning to discover protocols for quantum teleportation, entanglement distillation, and end-to-end bipartite entanglement distribution along a chain of quantum repeaters. While the work in \cite{WMDB19} is largely numerical, this work is focused on formally developing the mathematical tools needed to perform reinforcement learning of entanglement distribution protocols in general quantum networks. The development of the mathematical tools is essential when an agent acts in a quantum-mechanical environment, because it is important to understand how the agent's actions affect the quantum state of the environment. Furthermore, we expect that the protocols learned in \cite{WMDB19}, particularly those for entanglement distillation and entanglement swapping, could be incorporated as subroutines within the mathematical framework of decision processes developed in this work, so that large-scale quantum network protocols (going beyond the elementary link level) can be discovered using reinforcement learning.
	
	This work is also related to the work in \cite{DSM+19}, in which the authors develop a link-layer protocol for generating elementary links in a quantum network, and they perform simulations of entanglement distribution using a discrete-event simulator under various scenarios. The effect of different scheduling strategies is also considered. The protocols in \cite{DSM+19} consider actions in a more fine-grained manner than what we consider in this work. In particular, the steps required for heralding (namely, the communication signals for the results of the heralding) are explicitly taken into account. These steps can be incorporated within the framework developed here---all that has to be done is to appropriately define the transition maps (defined below) in order to accommodate the additional actions. We can similarly incorporate other classical discrete-valued properties of an elementary link into the elementary link status random variable $X_e(t)$ if needed.
	
	The approach to policy optimization taken in this work is similar to the approach in \cite{JTKL07}, in the sense that both approaches make use of the principle of dynamic programming. While in \cite{JTKL07} the focus is on obtaining end-to-end bipartite entanglement in a chain of quantum repeaters, the goal here is simply to examine elementary links and to determine the optimal sequence of actions that should be performed in order to maximize both the fidelity of an elementary link and the probability that an elementary link is active at any given time. Other recent work on optimization of quantum network protocols can be found in \cite{GEW20,STCMW20}.

\section{Elementary link generation}\label{sec-elem_link_generation}

	Let $e\in E$ denote an arbitrary (hyper)edge of the graph $G=(V,E)$ corresponding to the elementary links of a quantum network, and suppose that the edge contains $k\geq 2$ nodes $v_1,v_2,\dotsc,v_k$. The source associated with the elementary link corresponding to $e$ prepares a $k$-partite quantum state $\rho_{e}^S$ for $k$ quantum systems labeled $A_{e}^{v_1},A_{e}^{v_2},\dotsc,A_{e}^{v_k}$. This quantum state is distributed to the nodes by sending each quantum system $A_{e}^{v_i}$ through a quantum channel $\mathcal{S}_{e,i}$, so that the state after transmission is
	\begin{equation}
		\rho_{e}^{S,\text{out}}\coloneqq\left(\mathcal{S}_{e,1}\otimes\mathcal{S}_{e,2}\otimes\dotsb\otimes\mathcal{S}_{e,k}\right)(\rho_{e}^S)=\mathcal{S}_{e}(\rho_{e}^S),
	\end{equation}
	where $\mathcal{S}_{e}\coloneqq\mathcal{S}_{e,1}\otimes\mathcal{S}_{e,2}\otimes\dotsb\otimes\mathcal{S}_{e,k}$.
	
	After transmission from the source to the nodes, the nodes execute a \textit{heralding procedure}, which is an LOCC protocol executed by the nodes that confirms whether all of the nodes received their quantum systems. If the heralding procedure succeeds, then the nodes store their quantum systems in a quantum memory. Mathematically, the heralding procedure can be described by a quantum instrument $\{\mathcal{M}_{e}^0,\mathcal{M}_{e}^1\}$, which means that $\mathcal{M}_{e}^0$ and $\mathcal{M}_{e}^1$ are completely positive trace non-increasing maps such that $\mathcal{M}_{e}^0+\mathcal{M}_{e}^1$ is trace preserving. The map $\mathcal{M}_{e}^0$ corresponds to failure of the heralding procedure, and the map $\mathcal{M}_{e}^1$ corresponds to success. The outcome of the heralding procedure can then be captured by the following classical-quantum state:
	\begin{align}
		\widehat{\sigma}_{\!e}(1)&\coloneqq\ket{0}\bra{0}\otimes\mathcal{M}_{e}^0(\rho_{e}^{S,\text{out}})+\ket{1}\bra{1}\otimes\mathcal{M}_{e}^1(\rho_{e}^{S,\text{out}})\label{eq-link_cq_state_initial}\\
		&=\ket{0}\bra{0}\otimes \widetilde{\sigma}_{\!e}(0)+\ket{1}\bra{1}\otimes\widetilde{\sigma}_{\!e}(1),\label{eq-elem_link_initial_cq_state}
	\end{align}
	where the classical register holds the binary outcome of the heralding procedure (`1' for success and `0' for failure) and the quantum register holds the quantum state of the nodes corresponding to the outcome. In particular,
	\begin{equation}\label{eq-initial_state_tilde_0}
		\widetilde{\sigma}_{\!e}(0)\coloneqq(\mathcal{M}_{e}^0\circ\mathcal{S}_{e})(\rho_{e}^S)
	\end{equation}
	is the (unnormalized) quantum state corresponding to failure, and
	\begin{equation}\label{eq-initial_state_tilde_1}
		\widetilde{\sigma}_{\!e}(1)\coloneqq (\mathcal{M}_{e}^1\circ\mathcal{S}_{e})(\rho_{e}^{S})
	\end{equation}
	is the (unnormalized) quantum state corresponding to success. The (normalized) quantum states conditioned on success and failure, respectively, are defined to be
	\begin{equation}\label{eq-initial_link_states}
		\rho_{e}^0\coloneqq\frac{\widetilde{\sigma}_{\!e}(1)}{\Tr[\widetilde{\sigma}_{\!e}(1)]}, \quad \tau_{e}^{\varnothing}\coloneqq\frac{\widetilde{\sigma}_{\!e}(0)}{\Tr[\widetilde{\sigma}_{\!e}(0)]}.
	\end{equation}
	The superscript `0' in $\rho_{e}^0$ indicates that the quantum memories of the nodes are in their initial state immediately after success of the heralding procedure; we expand on this below. We let
	\begin{equation}\label{eq-elem_link_success_prob}
		p_{e}\coloneqq \Tr[\widetilde{\sigma}_{\!e}(1)]=\Tr[(\mathcal{M}_e^1\circ\mathcal{S}_e)(\rho_e^S)]
	\end{equation}
	denote the overall probability of success of the transmission from the source and of the heralding procedure.
	
	Now, as mentioned above, once the heralding procedure succeeds, the nodes store their quantum systems in their local quantum memory. We describe the decoherence of the quantum memories by quantum channels $\mathcal{N}_{e,i}$ acting on each quantum system $A_{e}^{v_i}$ of the elementary link corresponding to $e$, $i\in\{1,2,\dotsc,k\}$. The decoherence channel is applied at every time step in which the quantum system is in memory. The overall quantum channel acting on all of the quantum systems in the elementary link is
	\begin{equation}\label{eq-elem_link_decoherence_channel}
		\mathcal{N}_{e}\coloneqq\mathcal{N}_{e,1}\otimes\mathcal{N}_{e,2}\otimes\dotsb\otimes\mathcal{N}_{e,k}.
	\end{equation}
	The quantum state of the elementary link after $m$ time steps in the memories is therefore given by
	\begin{equation}
		\rho_{e}(m)\coloneqq\mathcal{N}_{e}^{\circ m}(\rho_{e}^0),
	\end{equation}
	where $\mathcal{N}_{e}^{\circ m}=\mathcal{N}_{e}\circ\mathcal{N}_{e}\circ\dotsb\circ\mathcal{N}_{e}$ ($m$ times). For a particular target/desired quantum state of the elementary link, which we assume to be a pure state $\psi=\ket{\psi}\bra{\psi}$, we let
	\begin{equation}
		f_{e}(m)\coloneqq \bra{\psi}\rho_{e}(m)\ket{\psi}=\bra{\psi}\mathcal{N}_{e}^{\circ m}(\rho_{e}^0)\ket{\psi}
	\end{equation}
	denote the fidelity of the state $\rho_e(m)$ with respect to the target state $\psi$.

\section{Details of the elementary link quantum decision process}\label{app-elem_link_QDP}

	The primary mathematical concept being used in this work is that of a Markov decision process. In particular, we consider a particular quantum generalization of a Markov decision process given in \cite{BBA14} (see also \cite{Cid16,YY18}), called a \textit{quantum partially observable Markov decision process}. For brevity, we use the term \textit{quantum decision process} throughout this work. In a quantum decision process, the agent's action at each time step results in a transformation of the quantum state of the environment, and the agent receives both partial (classical) information about the new quantum state of the environment along with a reward. In the context of elementary links in a quantum network, the elements of the quantum decision process that we formally define below are shown in Figure~\ref{fig-net_agent_env_0}. For a more general definition of a quantum decision process, we refer to \cite[Chapter~3]{Kha21}.
	
	\begin{definition}[Quantum decision process for elementary links]\label{def-network_QDP_elem_link}
		Let $G=(V,E)$ be the graph corresponding to the elementary links of a quantum network, and let $e\in E$. As shown in Figure~\ref{fig-net_agent_env_0}, we define a quantum decision process for $e$ by defining the agent for $e$ to be collectively the nodes belonging to $e$, and we define its environment to be the quantum systems distributed by the source station to the nodes of $e$. Then, the other elements of the quantum decision process are defined as follows.
		\begin{itemize}
			\item We denote the quantum systems of the environment collectively by $E^{e}$, and we let $E_t^{e}$ denote these quantum systems at time $t\geq 0$. We drop the superscript and simply write $E_t$ when the edge $e$ is understood from the context or unimportant in the context being considered. The quantum state of the environment at time $t=0$ is the source state $\rho_e^S\equiv\rho_{E_0^{e}}^{S}$.
			
			\item We let $\mathcal{X}=\{0,1\}$ tell us whether or not the elementary link is active at a particular time. From this, we define elementary link status random variables $X_{e}(t)$ for all $t\geq 1$ as follows:
				\begin{itemize}
					\item $X_{e}(t)=0$: elementary link is inactive (transmission and heralding not successful);
					\item $X_{e}(t)=1$: elementary link is active (transmission and heralding successful).
				\end{itemize}
				We let $\mathcal{A}=\{0,1\}$ be the set of possible actions of the agent, and we define corresponding action random variables $A_{e}(t)$ for all $t\geq 1$ as follows:
				\begin{itemize}
					\item $A_{e}(t)=0$: wait/keep the entangled state;
					\item $A_{e}(t)=1$: discard the entangled state and request a new entangled state.
				\end{itemize}
				The set of all histories up to time $t\geq 1$ is $(\mathcal{X}\times\mathcal{A})^{\times (t-1)}\times\mathcal{X}=\{0,1\}^{2t-1}$, and every element $h^t\in\{0,1\}^{2t-1}$ is a sequence of the form
				\begin{equation}
					h^t=(x_1,a_1,x_2,a_2,\dotsc,a_{t-1},x_t),
				\end{equation}
				with $x_j\in\{0,1\}$ being the elementary link status at time $j\in\{1,2,\dotsc, t\}$ and $a_j\in\{0,1\}$ being the action taken at time $j\in\{1,2,\dotsc,t-1\}$. The corresponding random variable for the history is
				\begin{equation}
					H_{e}(t)\coloneqq (X_{e}(1),A_{e}(1),X_{e}(2),A_{e}(2),\dotsc,A_{e}(t-1),X_{e}(t)).
				\end{equation}
				
				The random variables $X_{e}(t)$, $A_{e}(t)$, and $H_{e}(t)$ are mutually independent by definition for all~$e\in E$.
				
			\item The transformation of the quantum state of the environment from one time step to the next is given by a set of \textit{transition maps}. The transition maps tell us what effect each action has on the quantum state of the environment based on the status of the elementary link in the previous time step. We denote the transition maps by $\mathcal{T}_{e}^{x_t,a_t,x_{t+1}}\equiv \mathcal{T}_{E_t^{e}\to E_{t+1}^{e}}^{x_t,a_t,x_{t+1}}$ for all $x_t,a_t,x_{t+1}\in\{0,1\}$ and all $t\geq 1$, where
				\begin{align}
					\mathcal{T}_{e}^{x_t,1,1}(\sigma)&\coloneqq \Tr[\sigma](\mathcal{M}_{e}^1\circ\mathcal{S}_{e})(\rho_{e}^S)\quad\forall~x_t\in\{0,1\},\label{eq-network_QDP_trans_1}\\
					\mathcal{T}_{e}^{x_t,1,0}(\sigma)&\coloneqq \Tr[\sigma](\mathcal{M}_{e}^0\circ\mathcal{S}_{e})(\rho_{e}^S)\quad\forall~ x_t\in\{0,1\},\label{eq-network_QDP_trans_2}\\
					\mathcal{T}_{e}^{1,0,1}(\sigma)&\coloneqq \mathcal{N}_{e}(\sigma),\label{eq-network_QDP_trans_3}\\
					\mathcal{T}_{e}^{0,0,0}(\sigma)&\coloneqq \sigma,\label{eq-network_QDP_trans_4}
				\end{align}
				for all linear operators $\sigma$, where the definitions of the source transmission channel $\mathcal{S}_{e}$, the heralding quantum instrument $\{\mathcal{M}_{e}^0,\mathcal{M}_{e}^1\}$, and the decoherence channel $\mathcal{N}_{e}$ are given in Appendix~\ref{sec-elem_link_generation}. Superscript combinations not defined above are equal to the zero map by definition, i.e., $\mathcal{T}_{e}^{0,0,1}\coloneqq 0$ and $\mathcal{T}_{e}^{1,0,0}\coloneqq 0$. The transition maps are such that the sum $\sum_{x_{t+1}=0}^1 \mathcal{T}_e^{x_t,a_t,x_{t+1}}$ is a trace-preserving map for all $x_t,a_t\in\{0,1\}$. 
				
				The maps $\mathcal{T}_{e}^{0;x_1}\equiv\mathcal{T}_{E_0^{e}\to E_1^{e}}^{0;x_1}$, $x_1\in\{0,1\}$, are defined to be
				\begin{align}
					\mathcal{T}_{e}^{0;0}&\coloneqq\mathcal{M}_{e}^0\circ\mathcal{S}_{e},\label{eq-network_QDP_trans_5}\\
					\mathcal{T}_{e}^{0;1}&\coloneqq\mathcal{M}_{e}^1\circ\mathcal{S}_{e}.\label{eq-network_QDP_trans_6}
				\end{align}
				The sum $\mathcal{T}_e^{0;0}+\mathcal{T}_e^{0;1}$ is a trace-preserving map.
				
			\item Given a pure target state $\psi^{\text{target}}=\ket{\psi^{\text{target}}}\bra{\psi^{\text{target}}}$, the reward at time $t\geq 1$ is defined as follows:
				\begin{align}
					\mathcal{R}_{e}^{t;h^{t+1},1}(\cdot)&=\psi^{\text{target}}(\cdot)\psi^{\text{target}},\\
					\mathcal{R}_{e}^{t;h^{t+1},0}(\cdot)&=(\mathbbm{1}-\psi^{\text{target}})(\cdot)(\mathbbm{1}-\psi^{\text{target}}),
				\end{align}
				for all $h^{t+1}\in\{0,1\}^{2t+1}$, and we define functions $R_{e}(t):\{0,1\}^{2t+1}\times\{0,1\}\to\mathbb{R}$ as follows:
				\begin{align}
					R_{e}(t)(h^{t+1},0)&=0,\\
					R_{e}(t)(h^{t+1},1)&=\delta_{x_{t+1},1},
				\end{align}
				for every history $h^{t+1}=(x_1,a_1,\dotsc,x_t,a_t,x_{t+1})\in\{0,1\}^{2t+1}$.
				
			\item A $T$-step policy for the agent is a sequence of the form $\pi=(d_1,d_2,\dotsc,d_T)$, where the decision functions $d_t:\{0,1\}^{2t-1}\times\{0,1\}\to[0,1]$ are defined to be 
				\begin{equation}
					d_t(h^t)(a_t)\coloneqq\Pr[A_{e}(t)=a_t|H_{e}(t)=h^t]
				\end{equation}
				for all times $1\leq t\leq T$, histories $h^t\in\{0,1\}^{2t-1}$, and actions $a_t\in\{0,1\}$.~\defqed
		\end{itemize}
	\end{definition}
	
	Given an elementary link specified by an edge $e\in E$ and a policy $\pi=(d_1,d_2,\dotsc,d_t,\dotsc)$ for the corresponding agent, the agent-environment interaction, as depicted in Figure~\ref{fig-net_agent_env_0}, consists of a sequence of actions and responses of the agent and environment, respectively. This back-and-forth between the agent and the environment falls into the general paradigm of agent-environment interactions considered previously in \cite{DTB15,DTB16}, and more generally it falls within the theoretical framework of quantum causal networks (also referred to as quantum combs and quantum games) \cite{GW07,CDP08,CDP09,VW16}; see Figure~\ref{fig-QDP}.
	
	\begin{figure}
		\centering
		\includegraphics[scale=1.25]{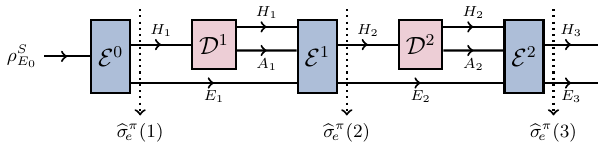}
		\caption{The agent-environment interaction corresponding to the quantum decision process for an elementary link in a quantum network, as given in Definition~\ref{def-network_QDP_elem_link}, can be viewed as a quantum causal network. The actions of the agent are given by \textit{decision channels}, defined in \eqref{eq-QDP_decision_channel} according to the agent's policy, and the corresponding changes in the quantum state of the environment are given by \textit{environment response channels}, defined in \eqref{eq-env_resp_chan_1} and \eqref{eq-env_resp_chan_2} according to the transition maps in \eqref{eq-network_QDP_trans_1}--\eqref{eq-network_QDP_trans_6}.}\label{fig-QDP}
	\end{figure}

	In terms of the quantum causal network shown in Figure~\ref{fig-QDP}, the actions of the agent are given by the \textit{decision channels}, which are defined as
	\begin{equation}\label{eq-QDP_decision_channel}
		\mathcal{D}_{H_t\to H_tA_t}^{t}(\ket{h^t}\bra{h^t}_{H_t})\coloneqq\ket{h^t}\bra{h^t}_{H_t}\otimes\sum_{a=0}^1d_t(h^t)(a)\ket{a}\bra{a}_{A_t},
	\end{equation}
	where for $h^t=(x_1,a_1,\dotsc,a_{t-1},x_t)$,
	\begin{equation}
		\ket{h^t}_{H_t}=\ket{x_1}_{X_1}\ket{a_1}_{A_1}\dotsb\ket{a_{t-1}}_{A_{t-1}}\ket{x_t}_{X_t}
	\end{equation}
	and $H_t\equiv X_1A_1\dotsb A_{t-1}X_t$, with $X_j$ and $A_j$ denoting classical registers that store the elementary link status and the action, respectively, at time $j$.
	
	The changes in the quantum state of the environment, based on the actions of the agent, are given by the \textit{environment response channels}, which are defined as
	\begin{align}
		&\mathcal{E}_{E_0\to H_1E_1}^0(\rho_{E_0})\coloneqq \sum_{x_1=0}^1\ket{x_1}\bra{x_1}_{X_1}\otimes\mathcal{T}_{E_0\to E_1}^{0;x_1}(\rho_{E_0}), \label{eq-env_resp_chan_1}\\
		&\mathcal{E}_{H_tA_tE_t\to H_{t+1}E_{t+1}}^t(\omega_{H_t}\otimes\sigma_{A_t}\otimes\rho_{E_t})\nonumber\\
		&\qquad\coloneqq \sum_{x_t,a_t,x_{t+1}=0}^1 \Tr_{X_tA_t}\!\left[(\omega_{H_t}\otimes\sigma_{A_t})\ket{x_t,a_t}\bra{x_t,a_t}_{X_tA_t}\right]\nonumber\\
		&\qquad\qquad\qquad\qquad\qquad\times\ket{x_t,a_t,x_{t+1}}\bra{x_t,a_t,x_{t+1}}_{X_tA_tX_{t+1}}\otimes\mathcal{T}_{E_t\to E_{t+1}}^{x_t,a_t,x_{t+1}}(\rho_{E_t}), \label{eq-env_resp_chan_2}
	\end{align}
	for arbitrary states $\rho_{E_0},\omega_{H_t},\sigma_{A_t},\rho_{E_t}$.
	
	In general, the classical-quantum state in \eqref{eq-elem_link_cq_state_QDP} is given by
	\begin{equation}
		\widehat{\sigma}_{\!e}^{\pi}(t)\equiv\widehat{\sigma}_{H_tE_t}^{\pi}(t)\coloneqq\mathcal{P}_{E_0\to H_tE_t}^{\pi;t}(\rho_{E_0}^S),
	\end{equation}
	where
	\begin{multline}
		\mathcal{P}_{E_0\to H_tE_t}^{\pi;t}\coloneqq\mathcal{E}_{H_{t-1}A_{t-1}E_{t-1}}^{t-1}\circ\mathcal{D}_{H_{t-1}\to H_{t-1}A_{t-1}}^{t-1}\circ\dotsb\\\circ\mathcal{E}_{H_1A_1E_1\to H_2E_2}^1\circ\mathcal{D}_{H_1\to H_1 A_1}^{1}\circ\mathcal{E}_{E_0\to H_1E_1}^0.
	\end{multline}
	Using the definitions of the decision and environment response channels, it is straightforward to show that
	\begin{equation}
		\widehat{\sigma}_{\!e}^{\pi}(t)=\sum_{h^t\in\{0,1\}^{2t-1}}\ket{h^t}\bra{h^t}_{H_t}\otimes\widetilde{\sigma}_{\!e}^{\pi}(t;h^t),
	\end{equation}
	where
	\begin{equation}\label{eq-cond_state_unnormalized}
		\widetilde{\sigma}_{\!e}^{\pi}(t;h^t)=\left(\prod_{j=1}^{t-1}d_j(h_j^t)(a_j)\right)\left(\mathcal{T}_e^{x_{t-1},a_{t-1},x_t}\circ\dotsb\circ\mathcal{T}_e^{x_2,a_2,x_3}\circ\mathcal{T}_e^{x_1,a_1,x_2}\circ\mathcal{T}_e^{0;x_1}\right)(\rho_e^S).
	\end{equation}
	Then, we have
	\begin{align}
		\Pr[H_e(t)=h^t]_{\pi}&=\Tr[\widetilde{\sigma}_{\!e}^{\pi}(t;h^t)],\\
		\sigma_{\!e}(t|h^t)&=\frac{\widetilde{\sigma}_{\!e}^{\pi}(t;h^t)}{\Pr[H_e(t)=h^t]_{\pi}}.
	\end{align}
	On the other hand, by the basic rules of probability, we have
	\begin{align}
		&\Pr[H_e(t)=h^t]_{\pi}=\Pr[X_e(1)=x_1,A_e(1)=a_1,\dotsc,A_e(t-1)=a_{t-1},X_e(t)=x_t]_{\pi}\\
		&\quad=\Pr[X_e(1)=x_1]\nonumber\\
		&\qquad\times \prod_{j=2}^t \left(\Pr[X_e(j)=x_j|H(j-1)=h_{j-1}^t,A(j-1)=a_{j-1}]\cdot d_{j-1}(h_{j-1}^t)(a_{j-1})\right).\label{eq-prob_history}
	\end{align}
	It follows that the transition probabilities in \eqref{eq-prob_history} are given by
	\begin{equation}
		\Pr[X_e(t+1)=x_{t+1}|H_e(t)=h^t,A_e(t)=a_t]=\Tr\!\left[\mathcal{T}_e^{x_t,a_t,x_{t+1}}(\sigma_{\!e}(t|h^t))\right].
	\end{equation}
	Using this, and the definition of the transition maps, we have the following values for the transition probabilities for all $t\geq 1$ and for every history $h^t=(x_1,a_1,\dotsc,a_{t-1},x_t)\in\{0,1\}^{2t-1}$:
	\begin{align}
		\Pr[X_{e}(t+1)=0|X_{e}(t)=x_t,A_{e}(t)=1]&=1-p_{e},\label{eq-link_trans_prob_1}\\
		\Pr[X_{e}(t+1)=1|X_{e}(t)=x_t,A_{e}(t)=1]&=p_{e},\label{eq-link_trans_prob_2}\\
		\Pr[X_{e}(t+1)=x_{t+1}|X_{e}(t)=x_t,A_{e}(t)=0]&=\delta_{x_t,x_{t+1}}\quad\forall~x_{t+1}\in\{0,1\},\label{eq-link_trans_prob_4}
	\end{align}
	where $p_e$ is the  success probability defined in \eqref{eq-elem_link_success_prob}. Observe that the transition probabilities are time independent. Furthermore, the status of an elementary link at time $t+1$ depends only on the status and action at time $t$, not on the entire history of statuses and actions. This reflects the fact that the transition maps also depend only on the status and action at the previous time step.
	
	For every edge $e\in E$ in the graph $G=(V,E)$ corresponding to the elementary links of a quantum network, and for every $T$-step policy $\pi$ for the elementary link corresponding to $e$, the expected reward at time $T$ is
	\begin{align}
		\mathbb{E}[R_e(T)]_{\pi}&\coloneqq\sum_{h^{T+1}\in\{0,1\}^{2T+1}}\sum_{s\in\{0,1\}}R_e(T)(h^{T+1},s)\Pr\left[H(T+1)=h^{T+1},s\right]_{\pi}\\
		&=\sum_{h^{T+1}\in\{0,1\}^{2T+1}}\sum_{s\in\{0,1\}}R_e(T)(h^{T+1},s) \Tr\left[\mathcal{R}_e^{T;h^{T+1},s}\left(\widetilde{\sigma}_{\!e}^{\pi}(T+1;h^{T+1})\right)\right]\\
		&=\sum_{\substack{h^{T+1}\in\{0,1\}^{2T+1}\\x_{T+1}=1}}\Tr\left[\psi^{\text{target}}\left(\widetilde{\sigma}_{\!e}^{\pi}(T+1;h^{T+1})\right)\right]\label{eq-QDP_exp_reward_formula_1}\\
		&=\Tr\left[(\ket{1}\bra{1}_{X_{T+1}}\otimes\psi^{\text{target}})\widehat{\sigma}_{\!e}^{\pi}(T+1)\right].\label{eq-QDP_exp_reward_formula_2}
	\end{align}
	In other words,
	\begin{equation}
		\mathbb{E}[R_e(T)]_{\pi}=\mathbb{E}[\widetilde{F}_{\!e}(T+1)]_{\pi}.
	\end{equation}

	Finally, the memory time random variable $M_e(t)$ defined in \eqref{eq-mem_time_def1} can be expressed in a closed form as follows:
	\begin{align}
		M_e(t)&=A_e(0)(X_e(1)+X_e(2)+\dotsb+X_e(t)-1)\overline{A_e(1)}\,\overline{A_e(2)}\dotsb\overline{A_e(t-1)}\nonumber\\
		&\quad +A_e(1)(X_e(2)+X_e(3)+\dotsb+X_e(t)-1)\overline{A_e(2)}\,\overline{A_e(3)}\dotsb\overline{A_e(t-1)}\nonumber\\
		&\quad +A_e(2)(X_e(3)+X_e(4)+\dotsb+X_e(t)-1)\overline{A_e(3)}\,\overline{A_e(4)}\dotsb\overline{A_e(t-1)}\nonumber\\
		&\quad +\dotsb\nonumber\\
		&\quad +A_e(t-1)(X_e(t)-1)\label{eq-mem_time_def2a}\\
		&=\sum_{j=1}^t A_e(j-1)\left(\sum_{\ell=j}^t X_e(\ell)-1\right)\prod_{k=j}^{t-1}\overline{A_e(k)}, \label{eq-mem_time_def2}
	\end{align}
	where $A_e(0)\equiv 1$ and $\overline{A_e(k)}\coloneqq 1-A(k)$ for all $k\geq 1$. This definition is indeed equivalent to the one in \eqref{eq-mem_time_def1}, because we can write the expression in \eqref{eq-mem_time_def2a} as
	\begin{equation}
		M_e(t)=\left(M_e(t-1)+X_e(t)\right)\overline{A_e(t-1)} + (X_e(t)-1)A_e(t-1),
	\end{equation}
	which is equivalent to \eqref{eq-mem_time_def1}.

\subsection{Proof of Theorem~\ref{thm-link_quantum_state}}\label{app-link_quantum_state_pf}

	First, let us observe that the statement of the proposition is true for $t=1$, because by \eqref{eq-network_QDP_trans_5}, \eqref{eq-network_QDP_trans_6}, and \eqref{eq-cond_state_unnormalized}, we can write
	\begin{equation}
		\widetilde{\sigma}_{\!e}^{\pi}(1;x_1)=x_1\widetilde{\rho}_{e}^0+(1-x_1)\widetilde{\tau}_{e}^{\varnothing},
	\end{equation}
	where $\widetilde{\rho}_{e}^0\coloneqq(\mathcal{M}_{e}^1\circ\mathcal{S}_{e})(\rho_{e}^S)$ and $\widetilde{\tau}_{e}^{\varnothing}\coloneqq(\mathcal{M}_{e}^0\circ\mathcal{S}_{e})(\rho_{e}^S)$. Then, indeed, we have $M_{e}(1)=0$ according to the definition in \eqref{eq-mem_time_def1}, as required, if $x_1=1$. Furthermore,
	\begin{equation}
		\Tr[\widetilde{\sigma}_{\!e}^{\pi}(1;x_1)]=x_1p_{e}+(1-x_1)(1-p_{e})=p_{e}^{x_1}(1-p_{e})^{1-x_1},
	\end{equation}
	so that
	\begin{align}
		\sigma_{\!e}(1|x_1)&=\frac{x_1\widetilde{\rho}_{e}^0+(1-x_1)\widetilde{\tau}_{e}^{\varnothing}}{p_{e}^{x_1}(1-p_{e})^{1-x_1}}\\
		&=\left\{\begin{array}{l l} \rho_{e}^0 & \text{if }x_1=1,\\ \tau_{e}^{\varnothing} & \text{if }x_1=0 \end{array}\right.\\
		&=x_1\rho_{e}^0+(1-x_1)\tau_{e}^{\varnothing},
	\end{align}
	where we recall the definitions of $\rho_{e}^0$ and $\tau_{e}^{\varnothing}$ from \eqref{eq-initial_link_states}.
	
	Now, for $t\geq 2$, we use \eqref{eq-cond_state_unnormalized}. Based on the definition of the transition maps, for every time step $j>1$ in which the action ``wait'' (i.e., $A_{e}(j)=0$) is performed and the elementary link is active (i.e., $X_{e}(j)=1$), the elementary link stays active at time step $j+1$, and thus by definition the memory time must be incremented by one, which is consistent with the definition of the memory time $M_{e}(t)$ given in \eqref{eq-mem_time_def1}, and the quantum state of the elementary link goes from $\rho_{e}(M_{e}(t))$ to $\rho_{e}(M_{e}(t)+1)$. If instead the elementary link is active at time $j$ and the action ``request'' is performed (i.e., $A_{e}(j)=1$), then the quantum state of the elementary link is discarded and is replaced either by the state $\rho_{e}^0$ (if $X_{e}(j+1)=1$) with probability $p_{e}$ or by the state $\tau_{e}^{\varnothing}$ (if $X_{e}(j+1)=0$) with probability $1-p_{e}$. In the former case, the memory time must be reset to zero, consistent with \eqref{eq-mem_time_def1}, and in the latter case, the memory time is $-1$, also consistent with \eqref{eq-mem_time_def1}.
	
	Furthermore, by definition of the transition maps, every time the action ``request'' is performed, we obtain a factor of $p_{e}$ (if the request succeeds) or $1-p_{e}$ (if the request fails). If the action ``wait'' is performed, then we obtain no additional multiplicative factors. The quantity $N_{e}^{\text{succ}}(t-1)$ is, by definition, equal to the number of requests that succeeded in $t-1$ time steps. Therefore, overall, we obtain a factor $p_{e}^{N_{e}^{\text{succ}}(t-1)}$ at the $(t-1)^{\text{st}}$ time step for the number of successful requests. The number of failed requests in $t-1$ time steps is given by
	\begin{align}
		\sum_{j=1}^{t-1} A_{e}(j-1)(1-X_{e}(j))&=\sum_{j=1}^{t-1} A_{e}(j-1)-\sum_{j=1}^{t-1} A_{e}(j-1)X_{e}(j)\\
		&=N_{e}^{\text{req}}(t-1)-N_{e}^{\text{succ}}(t-1),
	\end{align}
	so that we obtain an overall factor of $(1-p_{e})^{N_{e}^{\text{req}}(t-1)-N_{e}^{\text{succ}}(t-1)}$ at the $(t-1)^{\text{st}}$ time step for the failed requests. Also, the memory time at the $(t-1)^{\text{st}}$ time step is $M_{e}(t-1)(h_{t-1}^t)$, and then because the quantum state is either $\rho_{e}(M_{e}(t-1)(h_{t-1}^t))$ or $\tau_{e}^{\varnothing}$, we obtain
	\begin{align}
		\widetilde{\sigma}_{\!e}^{\pi}(t;h^t)&=\left(\prod_{j=1}^{t-1}d_j(h_j^t)(a_j)\right)p_{e}^{N_{e}^{\text{succ}}(t-1)(h_{t-1}^t)}(1-p_{e})^{N_{e}^{\text{req}}(t-1)(h_{t-1}^t)-N_{e}^{\text{succ}}(t-1)(h_{t-1}^t)}\nonumber\\
		&\qquad\qquad\times\left(x_{t-1}\mathcal{T}_{e}^{1,a_{t-1},x_t}(\rho_e(M_{e}(t-1)(h_{t-1}^t)))+(1-x_{t-1})\mathcal{T}_{e}^{0,a_{t-1},x_t}(\tau_{e}^{\varnothing})\right)\\
		&=\left(\prod_{j=1}^{t-1}d_j(h_j^t)(a_j)\right)p_{e}^{N_{e}^{\text{succ}}(t-1)(h_{t-1}^t)}(1-p_{e})^{N_{e}^{\text{req}}(t-1)(h_{t-1}^t)-N_{e}^{\text{succ}}(t-1)(h_{t-1}^t)}\nonumber\\
		&\qquad\qquad\times p_{e}^{a_{t-1}x_t}(1-p_{e})^{a_{t-1}(1-x_t)}(x_t\rho_{e}(M_{e}(t)(h^t))+(1-x_t)\tau_{e}^{\varnothing}) \\
		&=\left(\left(\prod_{j=1}^{t-1}d_j(h_j^t)(a_j)\right)p_{e}^{N_{e}^{\text{succ}}(t)(h^t)}(1-p_{e})^{N_{e}^{\text{req}}(t)(h^t)-N_{e}^{\text{succ}}(t)(h^t)}\right)\nonumber\\
		&\qquad\qquad\times (x_t\rho_{e}(M_{e}(t)(h^t))+(1-x_t)\tau_{e}^{\varnothing}).
	\end{align}
	Then, because $\Pr[H_{e}(t)=h^t]_{\pi}=\Tr[\widetilde{\sigma}_{\!e}^{\pi}(t;h^t)]$, we have
	\begin{equation}
		\Pr[H_{e}(t)=h^t]_{\pi}=\left(\prod_{j=1}^{t-1}d_j(h_j^t)(a_j)\right)p_{e}^{N_{e}^{\text{succ}}(t)(h^t)}(1-p_{e})^{N_{e}^{\text{req}}(t)(h^t)-N_{e}^{\text{succ}}(t)(h^t)},
	\end{equation}
	as required. Finally,
	\begin{equation}
		\sigma_{\!e}(t|h^t)=\frac{\widetilde{\sigma}_{\!e}^{\pi}(t;h^t)}{\Tr[\widetilde{\sigma}_{\!e}^{\pi}(t;h^t)]}=x_t\,\rho_{e}\!\left(M_{e}(t)(h^t)\right)+(1-x_t)\tau_{e}^{\varnothing},
	\end{equation}
	which completes the proof.

\subsection{Proof of Corollary~\ref{cor-link_avg_q_state}}\label{app-cor-link_avg_q_state_pf}

	Using the result of Theorem~\ref{thm-link_quantum_state}, the expected quantum state of the elementary link at time $t\geq 1$ is given by
	\begin{align}
		\sigma_{\!e}^{\pi}(t)&=\Tr_{H_t}[\widehat{\sigma}_{\!e}^{\pi}(t)]\\
		&=\sum_{h^t\in\{0,1\}^{2t-1}}\widetilde{\sigma}_{\!e}^{\pi}(t;h^t)\\
		&=\sum_{h^t\in\{0,1\}^{2t-1}}\Pr[H_{e}(t)=h^t]_{\pi}\left(X_{e}(t)(h^t)\rho_{e}\!\left(M_{e}(t)(h^t)\right)+(1-X_{e}(t)(h^t))\tau_{e}^{\varnothing}\right)\\
		&=\sum_{\substack{h^t\in\{0,1\}^{2t-1}:\\X_e(t)(h^t)=0}}\Pr[H_{e}(t)=h^t]_{\pi}\,\tau_{e}^{\varnothing}+\sum_{\substack{h^t\in\{0,1\}^{2t-1}:\\X_e(t)(h^t)=1}}\Pr[H_{e}(t)=h^t]_{\pi}\,\rho_{e}\!\left(M_{e}(t)(h^t)\right)\\
		&=(1-\Pr[X_{e}(t)=1]_{\pi})\tau_{e}^{\varnothing}+\sum_{m}\Pr[X_{e}(t)=1,M_{e}(t)=m]_{\pi}\,\rho_{e}(m),
	\end{align}
	where to obtain the last equality we used the fact that
	\begin{equation}
		\Pr[X_e(t)=0]_{\pi}=\sum_{\substack{h^t\in\{0,1\}^{2t-1}:\\X_e(t)(h^t)=0}}\Pr[H_e(t)=h^t]_{\pi}=1-\Pr[X_e(t)=1]_{\pi}.
	\end{equation}
	We also rearranged the sum with respect to the set $\{h^t\in\{0,1\}^{2t-1}:X_e(t)(h^t)=1\}$ so that the sum is with respect to the possible values of the memory time $m$, which in general depends on the policy $\pi$. This completes the proof.

\subsection{Proof of Theorem~\ref{thm-succ_prob_fid_cq}}\label{app-thm-succ_prob_fid_cq_pf}

	To see the first equality in \eqref{eq-link_value_prob_via_state}, observe that
	\begin{equation}
		\Tr\left[\ket{1}\bra{1}_{X_t}\widehat{\sigma}_{\!e}^{\pi}(t)\right]=\sum_{\substack{h^t\in\{0,1\}^{2t-1}:\\X_{e}(t)(h^t)=1}}\Pr[H_{e}(t)=h^t]_{\pi}.
	\end{equation}
	The expression on the right-hand side of this equation is equal to $\Pr[X_{e}(t)=1]_{\pi}$ by definition of the random variable $X_{e}(t)$. The second equality in \eqref{eq-link_value_prob_via_state} holds because $X_{e}(t)$ is a binary/Bernoulli random variable.
	
	To see \eqref{eq-avg_Ftilde}, we first use the definition of expectation to get
	\begin{equation}
		\mathbb{E}[\widetilde{F}_{\!e}(t)]_{\pi}=\sum_{m} f_e(m)\Pr[X_{e}(t)=1,M_{e}(t)=m]_{\pi},
	\end{equation}
	where the sum is with respect to all possible values of the random variable $M_{e}(t)$, which depends on the policy $\pi$. Then, by Theorem~\ref{thm-link_quantum_state},
	\begin{align}
		\Tr\left[\left(\ket{1}\bra{1}_{X_{t}}\otimes\psi\right)\widehat{\sigma}_{\!e}^{\pi}(t)\right]&=\sum_{\substack{h^{t}\in\{0,1\}^{2t-1}:\\X_e(t)(h^t)=1}}\Pr[H_{e}(t)=h^{t}]_{\pi}\bra{\psi}\rho_e(M_{e}(t)(h^{t}))\ket{\psi}\\
		&=\sum_{\substack{h^{t}\in\{0,1\}^{2t-1}:\\X_e(t)(h^t)=1}}\Pr[H_{e}(t)=h^{t}]_{\pi}f_e(M_{e}(t)(h^{t}))\\
		&=\sum_m f_e(m)\Pr[X_{e}(t)=1,M_{e}(t)=m]_{\pi},
	\end{align}
	where the last equality holds because the sum with respect to the set $\{h^{t}\in\{0,1\}^{2t-1}:X_e(t)(h^t)=1\}$ can be rearranged into a sum with respect to the possible values of the memory time $M_{e}(t)$ when the elementary link is active. This completes the proof.

\subsection{Proof of Theorem~\ref{thm-opt_policy}}\label{app-elem_link_QDP_alg}

	We start with a lemma.
	
	\begin{lemma}\label{lem-QDP_policy_eval}
		Let $G=(V,E)$ be the graph corresponding to the elementary links of a quantum network, let $e\in E$, let $\psi\equiv\ket{\psi}\bra{\psi}$ be a pure target state, and let $\pi=(d_1,d_2,\dotsc,d_T)$ be a $T$-step policy with $T\geq 1$. Then,
		\begin{equation}
			\mathbb{E}[\widetilde{F}_{\!e}(T+1)]_{\pi}=\sum_{x_1,a_1=0}^1 d_1(x_1)(a_1)\,v_2^{\pi}(x_1,a_1),
		\end{equation}
		where
		\begin{equation}\label{eq-QDP_policy_eval_1}
			v_t^{\pi}(h^{t-1},a_{t-1})\coloneqq\sum_{x_t,a_t=0}^1 d_t(h^{t-1},a_{t-1},x_t)(a_t)\,v_{t+1}^{\pi}(h^{t-1},a_{t-1},x_t,a_t)
		\end{equation}
		for all $2\leq t\leq T$, $h^{t-1}\in\{0,1\}^{2t-3}$, and $a_{t-1}\in\{0,1\}$, and
		\begin{equation}\label{eq-QDP_policy_eval_2}
			v_{T+1}^{\pi}(h^T,a_T)\coloneqq p_{e}^{y}(1-p_{e})^{x-y} f_e\!\left(M_e(T+1)(h^T,a_T,1)\right)
		\end{equation}
		for all $h^T\in\{0,1\}^{2T-1}$ and $a_T\in\{0,1\}$, where $y=N_{e}^{\text{succ}}(T+1)(h^T,a_T,1)$, $x=N_{e}^{\text{req}}(T+1)(h^T,a_T,1)$.~\defqed
	\end{lemma}
	\smallskip
	\begin{remark}	
		The functions $v_t^{\pi}$ that we have defined in the statement of the lemma can be thought of as analogous to \textit{action-value functions} in classical Markov decision processes; see, e.g., \cite{Put14_book,Sut18_book}. Also, observe that \eqref{eq-QDP_policy_eval_1} and \eqref{eq-QDP_policy_eval_2} specify a \textit{backward recursion algorithm} for evaluating a given policy. The algorithm proceeds by first evaluating the function $v_{T+1}^{\pi}$, then proceeding backwards, calculating $v_t^{\pi}$ for all $T\geq t\geq 2$ in order to finally obtain $\mathbb{E}[\widetilde{F}_{\!e}(T+1)]_{\pi}$.~\defqed
	\end{remark}
	
	\begin{proof}[Proof of Lemma~\ref{lem-QDP_policy_eval}]
		Using \eqref{eq-QDP_exp_reward_formula_1} and \eqref{eq-cond_state_unnormalized}, we have that
		\begin{align}
			\mathbb{E}[\widetilde{F}_{\!e}(T+1)]_{\pi}&=\sum_{\substack{h^{T+1}\in\{0,1\}^{2T+1}:\\x_{T+1}=1}}\Tr\left[\psi\left(\widetilde{\sigma}_{\!e}^{\pi}(T+1;h^{T+1})\right)\right]\\
			&=\sum_{\substack{h^{T+1}\in\{0,1\}^{2T+1}:\\x_{T+1}=1}}\left(\prod_{t=1}^T d_t(h_t^{T+1})(a_t)\right)\Tr\left[\psi\left(\widetilde{\sigma}_{\!e}^{(\mathsf{E})}(T+1;h^{T+1})\right)\right].\label{eq-QDP_policy_eval_pf1}
		\end{align}
		For brevity, we denote histories $h^{T+1}=(x_1,a_1,\dotsc,x_T,a_T,x_{T+1})$ such that $x_{T+1}=1$ by $(h^T,a_T,1)$, where $h^T=(x_1,a_1,\dotsc,a_{T-1},x_T)$. Then,
		\begin{align}
			\widetilde{\sigma}_{\!e}^{(\mathsf{E})}(T+1;h^{T+1})&\equiv\widetilde{\sigma}_{\!e}^{(\mathsf{E})}(T+1;h^T,a_T,1)\\
			&\coloneqq\left(\mathcal{T}_e^{x_T,a_T,1}\circ\dotsb\circ\mathcal{T}_e^{x_2,a_2,x_3}\circ\mathcal{T}_e^{x_1,a_1,x_2}\circ\mathcal{T}_e^{0;x_1}\right)(\rho_e^S). \label{eq-cond_state_unnormalized_env}
		\end{align}
		From this, we see that
		\begin{equation}
			\mathbb{E}[\widetilde{F}_{\!e}(T+1)]_{\pi}=\sum_{x_1,a_1=0}^1 d_1(x_1)(a_1)\, v_2^{\pi}(x_1,a_1),
		\end{equation}
		where
		\begin{equation}
			v_2^{\pi}(x_1,a_1)\coloneqq\sum_{\substack{x_2,\dotsc,x_T\in\{0,1\}\\a_2,\dotsc,a_T\in\{0,1\}}}\left(\prod_{t=2}^T d_t(h_t^{T+1})(a_t)\right)\Tr\left[\psi\left(\widetilde{\sigma}_{\!e}^{(\mathsf{E})}(T+1;h^T,a_T,1)\right)\right].
		\end{equation}
		Then, separating the sum with respect to $x_2,a_2\in\{0,1\}$ in the above equation leads to
		\begin{equation}
			v_2^{\pi}(x_1,a_1)=\sum_{x_2,a_2=0}^1 d_2(h_2^{T+1})(a_2)\, v_3^{\pi}(\underbrace{x_1,a_1,x_2}_{h^{T+1}_2},a_2),
		\end{equation}
		where
		\begin{equation}
			v_3^{\pi}(h^{T+1}_2,a_2)\coloneqq\sum_{\substack{x_3,\dotsc,x_T\in\{0,1\}\\a_3,\dotsc,a_T\in\{0,1\}}}\left(\prod_{t=3}^T d_t(h_t^{T+1})(a_t)\right)\Tr\left[\psi\left(\widetilde{\sigma}_{\!e}^{(\mathsf{E})}(T+1;h^T,a_T,1)\right)\right].
		\end{equation}
		Proceeding in this manner, we define functions $v_t^{\pi}(h^{t-1},a_{t-1})$ for $2\leq t\leq T$ as follows:
		\begin{equation}
			v_t^{\pi}(h^{t-1},a_{t-1})\coloneqq\sum_{x_t,a_t=0}^1 d_t(h^{t-1},a_{t-1},x_t)(a_t)\, v_{t+1}^{\pi}(h^{t-1},a_{t-1},x_t,a_t).
		\end{equation}
		In particular, for $t=T$, we have
		\begin{align}
			v_T^{\pi}(h^{T-1},a_{T-1})&\coloneqq\sum_{x_T,a_T\in\{0,1\}} d_T(h^{T-1},a_{T-1},x_T)(a_T) \Tr\left[\psi\left(\widetilde{\sigma}_{\!e}^{(\mathsf{E})}(T+1;h^T,a_T,1)\right)\right]\\
			&=\sum_{x_T,a_T=0}^1 d_T(h^{T-1},a_{T-1},x_T)(a_T)\, v_{T+1}^{\pi}(h^T,a_T),
		\end{align}
		where
		\begin{equation}
			v_{T+1}^{\pi}(h^T,a_T)\coloneqq\bra{\psi}\widetilde{\sigma}_{\!e}^{(\mathsf{E})}(T+1;h^T,a_T,1)\ket{\psi}.
		\end{equation}
		Now, it follows from Theorem~\ref{thm-link_quantum_state} that
		\begin{equation}
			\widetilde{\sigma}_{\!e}^{(\mathsf{E})}(T+1;h^T,a_T,1)=p_{e}^{y}(1-p_{e})^{x-y} \rho_e\!\left(M_e(T+1)(h^T,a_T,1)\right),
		\end{equation}
		where $y=N_{e}^{\text{succ}}(T+1)(h^T,a_T,1)$ and $x=N_{e}^{\text{req}}(T+1)(h^T,a_T,1)$. Therefore,
		\begin{equation}
			\bra{\psi}\widetilde{\sigma}_{\!e}^{(\mathsf{E})}(T+1;h^T,a_T,1)\ket{\psi}=p_{e}^{y}(1-p_{e})^{x-y} f_e\!\left(M_e(T+1)(h^T,a_T,1)\right).
		\end{equation}
		This completes the proof.
	\end{proof}
	\smallskip
	\begin{remark}\label{rem-QDP_policy_eval_2}
		There is an advantage to using the backward recursion algorithm (as presented in Lemma~\ref{lem-QDP_policy_eval}) to evaluate a policy, rather than simply using the definition of the expected fidelity in \eqref{eq-QDP_exp_reward_formula_1} or \eqref{eq-QDP_exp_reward_formula_2}. This advantage comes from the fact that the function $v_{T+1}^{\pi}$ defined in \eqref{eq-QDP_policy_eval_2} is independent of the policy $\pi$---it depends only on the elements of the environment and on the horizon time. Therefore, for a given elementary link specified by the edge $e$, and a given horizon time $T$, the function values $v_{T+1}^{\pi}(h^T,a_T)\equiv v_{T+1}(h^T,a_T)$ can be computed once and need never be computed again. Then, given a $T$-step policy $\pi$, the backward recursion algorithm can be used to quickly evaluate the expected fidelity.~\defqed
	\end{remark}
	
	\begin{proof}[Proof of Theorem~\ref{thm-opt_policy}]
		We start with the backward recursion algorithm given in Lemma~\ref{lem-QDP_policy_eval}. Let
		\begin{equation}
			\pi=\left(d_t(h^t):1\leq t\leq T,\,h^t\in\{0,1\}^{2t-1}\right)
		\end{equation}
		denote an arbitrary $T$-step policy, and let us define
		\begin{equation}
			\pi^{(t)}\coloneqq\left(d_j(h^j):t\leq j\leq T,\,h^j\in\{0,1\}^{2j-1}\right)
		\end{equation}
		to be the ``slices'' of $\pi$ from time $t$ onwards. By observing that $v_t^{\pi}$ depends only on the policy from time $t$ onwards, i.e., on $\pi^{(t)}$, we find that
		\begin{equation}
			\max_{\pi}\mathbb{E}[\widetilde{F}_{\!e}(T+1)]_{\pi}=\max_{d_1}\sum_{x_1,a_1=0}^1 d_1(x_1)(a_1)\max_{\pi^{(2)}}v_2^{\pi^{(2)}}(x_1,a_1),
		\end{equation}
		Then, using \eqref{eq-QDP_policy_eval_1}, we have that
		\begin{multline}
			\max_{\pi^{(t)}}v_t^{\pi^{(t)}}(h^{t-1},a_{t-1})\\=\max_{d_t}\sum_{x_t,a_t=0}^1 d_t(h^{t-1},a_{t-1},x_t)(a_t)\max_{\pi^{(t+1)}}v_{t+1}^{\pi^{(t+1)}}(h^{t-1},a_{t-1},x_t,a_t)
		\end{multline}
		for all $2\leq t\leq T$, $h^{t-1}\in\{0,1\}^{2t-3}$, and $a_{t-1}\in\{0,1\}$. By defining the functions
		\begin{align}
			w_t&\coloneqq\max_{\pi^{(t)}}v_t^{\pi^{(t)}}\quad\forall~2\leq t\leq T,\\
			w_{T+1}&\coloneqq v_{T+1}^{\pi}
		\end{align}
		(recall that $v_{T+1}^{\pi}$ does not depend on $\pi$; see Remark~\ref{rem-QDP_policy_eval_2}), we see that the optimization problem reduces to the following:
		\begin{multline}\label{eq-QDP_opt_pol_pf1}
			\max_{\pi}\mathbb{E}[\widetilde{F}_{\!e}(T+1)]_{\pi}\\=\max\left\{\sum_{x_1,a_1=0}^1\Tr\left[M_{A_1}^{1;a_1}\rho_{A_1}^{x_1}\right] w_2(x_1,a_1): \rho_{A_1}^{x_1}\geq 0,\,\Tr[\rho_{A_1}^{x_1}]=1,\,x_1\in\{0,1\}\right\},
		\end{multline}
		where $M_{A_1}^{1;a_1}=\ket{a_1}\bra{a_1}$ and
		\begin{multline}\label{eq-QDP_opt_pol_pf2}
			w_t(h^{t-1},a_{t-1})\\\coloneqq\max\left\{\sum_{x_t,a_t=0}^1\Tr\left[M_{A_t}^{t;a_t}\rho_{A_t}^{x_t}\right]w_{t+1}(h^t,a_t):\rho_{A_t}^{x_t}\geq 0,\,\Tr[\rho_{A_t}^{x_t}]=1,\,x_t\in\{0,1\}\right\}
		\end{multline}
		for all $2\leq t\leq T$, $h^{t-1}\in\{0,1\}^{2t-3}$, and $a_{t-1}\in\{0,1\}$, with $M_{A_t}^{t;a_t}=\ket{a_t}\bra{a_t}$.
		
		Now, observe that the objective function in \eqref{eq-QDP_opt_pol_pf1} can be written as
		\begin{align}
			\sum_{x_1,a_1=0}^1\Tr\left[M_{A_1}^{1;a_1}\rho_{A_1}^{x_1}\right]w_2(x_1,a_1)&=\sum_{x_1=0}^1\Tr\left[\left(\sum_{a_1=0}^1 M_{A_1}^{1;a_1}w_2(x_1,a_1)\right)\rho_{A_1}^{x_1}\right]\\
			&=\sum_{x_1=0}^1\Tr\left[\widetilde{M}_{A_1}^{1;x_1}\rho_{A_1}^{x_1}\right],
		\end{align}
		where in the last line we defined $\widetilde{M}_{A_1}^{t;x_1}\coloneqq\sum_{a_1=0}^1 M_{A_1}^{t;a_1}w_2(x_1,a_1)$. Similarly, in \eqref{eq-QDP_opt_pol_pf2}, the objective function can be written as
		\begin{align}
			\sum_{x_t,a_t=0}^1\Tr\left[M_{A_t}^{t;a_t}\rho_{A_t}^{x_t}\right]w_{t+1}(h^t,a_t)=\sum_{x_t=0}^1\Tr\left[\widetilde{M}_{A_t}^{t;h^t}\rho_{A_t}^{x_t}\right],
		\end{align}
		where $\widetilde{M}_{A_t}^{t;h^t}\coloneqq\sum_{a_t=0}^1M_{A_t}^{t;a_t}w_{t+1}(h^t,a_t)$. Therefore, we have
		\begin{align}
			\max_{\pi}\mathbb{E}[\widetilde{F}_{\!e}(T+1)]_{\pi}&=\max\left\{\sum_{x_1=0}^1\Tr\left[\widetilde{M}_{A_1}^{1;x_1}\rho_{A_1}^{x_1}\right]:\rho_{A_1}^{x_1}\geq 0,\,\Tr[\rho_{A_1}^{x_1}]=1,\,x_1\in\{0,1\}\right\},\label{eq-QDP_opt_pol_pf3}\\
			w_t(h^{t-1},a_{t-1})&=\max\left\{\sum_{x_t=0}^1\Tr\left[\widetilde{M}_{A_t}^{t;h^t}\rho_{A_t}^{x_t}\right]:\rho_{A_t}^{x_t}\geq 0,\,\Tr[\rho_{A_t}^{x_t}]=1,\,x_t\in\{0,1\}\right\},\label{eq-QDP_opt_pol_pf4}
		\end{align}
		for all $2\leq t\leq T$, $h^{t-1}\in\{0,1\}^{2t-3}$, and $a_{t-1}\in\{0,1\}$. Now, notice that for every time step the optimization problem is of the form
		\begin{equation}\label{eq-opt_reward_generic}
			\begin{array}{l l} \text{maximize} & \displaystyle \sum_{s\in\mathcal{S}}\Tr[H^s\rho^s] \\[0.6cm] \text{subject to} & \rho^s\geq 0,\,\Tr[\rho^s]=1\quad\forall~s\in\mathcal{S}, \end{array}
		\end{equation}
		where $\mathcal{S}$ is some finite set and $\{H^s\}_{s\in\mathcal{S}}$ is some set of Hermitian operators. Because the optimization is with respect to the independent variables $\{\rho^s\}_{s\in\mathcal{S}}$, the maximum can be brought inside the sum, so that the solution to the optimization problem is simply
		\begin{equation}
			\sum_{s\in\mathcal{S}}\lambda_{\max}(H^s),
		\end{equation}
		where we have used to fact that, for every Hermitian operator $H$,
		\begin{equation}\label{eq-max_EV_SDP}
			\max_{\rho:\rho\geq 0,\Tr[\rho]=1}\Tr[H\rho]=\lambda_{\max}(H),
		\end{equation}
		where $\lambda_{\max}(H)$ denotes the largest eigenvalue of $H$ \cite{BV96}. A state $\rho$ achieving the maximum is $\ket{\lambda_{\max}(H)}\bra{\lambda_{\max}(H)}$, where $\ket{\lambda_{\max}(H)}$ is an eigenvector corresponding to the largest eigenvalue of $H$. Applying this result to \eqref{eq-QDP_opt_pol_pf3} and \eqref{eq-QDP_opt_pol_pf4} leads to the desired result, because
		\begin{equation}\label{eq-QDP_opt_pol_pf5}
			\lambda_{\max}(\widetilde{M}_{A_t}^{t;h^t})=\lambda_{\max}\left(\sum_{a_t=0}^1 w_{t+1}(h^t,a_t)\ket{a_t}\bra{a_t}\right)=\max_{a_t\in\{0,1\}}w_{t+1}(h^t,a_t).
		\end{equation}
		Note that the optimal action at the $t^{\text{th}}$ time step for the history $h^t\in\{0,1\}^{2t-1}$ is given by the value $a_t$ that achieves the maximum in \eqref{eq-QDP_opt_pol_pf5}, which gives us the result in \eqref{eq-backward_recursion_policy}. This completes the proof.
	\end{proof}

\subsection{Proof of Theorem~\ref{thm-network_QDP_forward_recursion_policy}}\label{app-network_QDP_forward_recursion_policy_pf}

	Let $t\geq 1$, and consider a policy $\pi$ up to time $t$. With respect to this policy, the classical-quantum state of the elementary link is (recall \eqref{eq-elem_link_cq_state_QDP})
	\begin{equation}
		\widehat{\sigma}_{\!e}^{\pi}(t)=\sum_{h^t\in\{0,1\}^{2t-1}}\ket{h^t}\bra{h^t}\otimes\widetilde{\sigma}_{\!e}^{\pi}(t;h^t).
	\end{equation}
	Now, for an arbitrary decision function $d_t$ corresponding to the decision at time $t$, we obtain the following classical-quantum state of the elementary link at time $t+1$:
	\begin{align}
		\widehat{\sigma}_{\!e}^{(\pi,d_t)}(t+1)&=\sum_{h^{t+1}\in\{0,1\}^{2t+1}}\ket{h^{t+1}}\bra{h^{t+1}}\otimes\widetilde{\sigma}_{\!e}^{(\pi,d_t)}(t+1;h^{t+1})\\
		&=\sum_{\substack{h^t\in\{0,1\}^{2t-1},\\a_t,x_{t+1}\in\{0,1\}}}\ket{h^t,a_t,x_{t+1}}\bra{h^t,a_t,x_{t+1}}\otimes d_t(h^t)(a_t)\mathcal{T}_{e}^{x_t,a_t,x_{t+1}}\left(\widetilde{\sigma}_{\!e}^{\pi}(t;h^t)\right).
	\end{align}
	Then,
	\begin{align}
		\mathbb{E}[\widetilde{F}_{\!e}(t+1)]_{(\pi,d_t)}&=\Tr\left[\left(\ket{1}\bra{1}_{X_{t+1}}\otimes\psi\right)\widehat{\sigma}_{\!e}^{(\pi,d_t)}(t+1)\right]\\
		&=\sum_{h^t\in\{0,1\}^{2t-1}}\sum_{a_t=0}^1 d_t(h^t)(a_t)\bra{\psi}\mathcal{T}_e^{x_t,a_t,1}(\widetilde{\sigma}_{\!e}^{\pi}(t;h^t))\ket{\psi}\\
		&=\sum_{h^t\in\{0,1\}^{2t-1}}\Tr[N^{h^t}\rho^{h^t}],
	\end{align}
	where
	\begin{align}
		N^{h^t}&\coloneqq\sum_{a_t=0}^1 \bra{\psi}\mathcal{T}_e^{x_t,a_t,1}(\widetilde{\sigma}_{\!e}^{\pi}(t;h^t))\ket{\psi}\,\ket{a_t}\bra{a_t},\\
		\rho^{h^t}&\coloneqq\sum_{a_t=0}^1d_t(h^t)(a_t)\ket{a_t}\bra{a_t}.
	\end{align}
	So we have that
	\begin{equation}
		\max_{d_t}\mathbb{E}[\widetilde{F}_{\!e}(t+1)]_{(\pi,d_t)}=\max_{\{\rho^{h^t}\}_{h^t}}\sum_{h^t\in\{0,1\}^{2t-1}}\Tr[N^{h^t}\rho^{h^t}],
	\end{equation}
	which is an optimization problem of the form in \eqref{eq-opt_reward_generic}. The optimal value is therefore equal to $\sum_{h^t\in\{0,1\}^{2t-1}}\lambda_{\max}(N^{h^t})$, with associated optimal decision function, which we denote by $d_t^{\text{FR}}$, equal to
	\begin{equation}\label{eq-forward_recursion_pf}
		d_t^{\text{FR}}(h^t)=\argmax_{a_t\in\{0,1\}}\bra{\psi}\mathcal{T}_e^{x_t,a_t,1}(\widetilde{\sigma}_{\!e}^{\pi}(t;h^t))\ket{\psi}\quad\forall\,h^t\in\{0,1\}^{2t-1}.
	\end{equation}
	Now,
	\begin{align}
		a_t=0\Rightarrow \bra{\psi}\mathcal{T}_{e}^{x_t,0,1}\left(\widetilde{\sigma}_{\!e}^{\pi}(t;h^t)\right)\ket{\psi}&=\Pr[H_e(t)=h^t]_{\pi}x_t\bra{\psi}\mathcal{N}_{e}\left(\rho_{e}\!\left(M_{e}(t)(h^t)\right)\right)\ket{\psi}\nonumber\\[0.2cm]
		&=\Pr[H_e(t)=h^t]_{\pi}x_t f_e(M_{e}(t)(h^t)+1),\\[0.3cm]
		a_t=1\Rightarrow \bra{\psi}\mathcal{T}_{e}^{x_t,1,1}\left(\widetilde{\sigma}_{\!e}^{\pi}(t;h^t)\right)\ket{\psi}&=\Pr[H_e(t)=h^t]_{\pi}p_{e} f_e(0).
	\end{align}
	So the task is to determine which of the two quantities, $x_tf_e(M_{e}(t)(h^t)+1)$ and $p_{e}f_e(0)$, is higher, where $p_e$ is the  success probability. If the elementary link is not active at time $t$, meaning that $x_t=0$, then requesting a link, i.e., selecting $a_t=1$, gives a higher value than selecting $a_t=0$ (because the latter leads to a value of zero for the objective function in \eqref{eq-forward_recursion_pf} for all $p_e>0$). On the other hand, if the elementary link is active at time $t$, then the task is to compare $f_e(M_{e}(t)(h^t)+1)$ and $p_{e}f_e(0)$ for every history $h^t\in\{0,1\}^{2t-1}$. Which of these two quantities is higher (and thus which action is taken) depends on the success probability $p_{e}\in(0,1)$, the noise model of the quantum memory, and on the target pure state $\psi$. We conclude that the decision function $d_t^{\text{FR}}$ in \eqref{eq-forward_recursion_pf} is given by \eqref{eq-network_QDP_forward_recursion_policy}, as claimed.

\subsection{Other figures of merit}\label{sec-other_figs_of_merit}
	
	In addition to the figures of merit defined in Section~\ref{sec-figures_of_merit_general}, there are two other figures of merit of interest that we can consider.

	\begin{definition}[Figures of merit for an elementary link policy]\label{def-fig_of_merit_elem_link}
		Let $G=(V,E)$ be the graph corresponding to the elementary links of a quantum network, let $e\in E$, and let $t\geq 1$. Given a policy $\pi$ for the elementary link corresponding to $e$, we define the following figures of merit for the policy $\pi$.
		\begin{itemize}		
			\item The expected waiting time for the elementary link to become active, i.e., $\mathbb{E}[W_e(t_{\text{req}})]_{\pi}$, where
				\begin{equation}
					W_e(t_{\text{req}})\coloneqq\sum_{t=t_{\text{req}}+1}^{\infty} tX_{e}(t)\prod_{j=t_{\text{req}}+1}^{t-1}(1-X_e(j)),
				\end{equation}
				and $t_{\text{req}}\geq 0$ is the time at which the initial request for the elementary link is made. In particular,
				\begin{equation}\label{eq-waiting_time_prob_late_request}
					\Pr[W_e(t_{\text{req}})=t]_{\pi}=\Pr[X_e(t_{\text{req}}+1)=0,\dotsc,X_e(t_{\text{req}}+t)=1]_{\pi}.
				\end{equation}
			
			\item The expected success rate of the elementary link, i.e., $\mathbb{E}[S_{\!e}(t)]_{\pi}$, where
				\begin{equation}
					S_{\!e}(t)\coloneqq\frac{\displaystyle\sum_{j=1}^t A_e(j-1)X_e(j)}{\displaystyle\sum_{j=1}^t A_e(j-1)}.
				\end{equation}
				The success rate is simply the ratio of the number of successful transmissions when a request is made to the total number of requests made within time $t$. We let $A_e(0)\equiv 1$.~\defqed
		\end{itemize}	
	\end{definition}

\section{Details of the memory-cutoff policy}\label{app-mem_cutoff_details}

	In this appendix, we go through the details of the memory-cutoff policy. We start by making the following initial remarks about the definition of the memory-cutoff policy.
	\begin{itemize}
		
		\item Observe that for $t^{\star}=\infty$, we can write the decision function $d_t^{\infty}$ in the following simpler form:
			\begin{equation}\label{eq-mem_cutoff_infty_dt_simpler}
				d_t^{\infty}(h^t)=\left\{\begin{array}{l l} 0 & \text{if }X_e(t)(h^t)=1, \\ 1 & \text{if } X_e(t)(h^t)=0. \end{array}\right.
			\end{equation}
			In other words, for the $t^{\star}=\infty$ memory-cutoff policy, it suffices to look at the status of the elementary link at the current time in order to determine the next action.
		
		\item We denote the probability distribution of the history random variable $H_e(t)$ with respect to the $t^{\star}$ memory-cutoff policy by $\Pr[H_e(t)=h^t]_{t^{\star}}$, and similarly for the marginal distributions.
		
		\item It is straightforward to see that the following conditional probabilities hold for all $t^{\star}\in\mathbbm{N}$:
			\begin{align}
				&\Pr[X_e(t+1)=1,M_e(t+1)=0|X_e(t)=0,M_e(t)=-1]=p_e,\label{eq-cutoff_trans_prob1}\\
				&\Pr[X_e(t+1)=1,M_e(t+1)=0|X_e(t)=1,M_e(t)=t^{\star}]=p_e,\label{eq-cutoff_trans_prob2}\\
				&\Pr[X_e(t+1)=0,M_e(t+1)=-1|X_e(t)=0,M_e(t)=-1]=1-p_e,\label{eq-cutoff_trans_prob3}\\
				&\Pr[X_e(t+1)=0,M_e(t+1)=-1|X_e(t)=1,M_e(t)=t^{\star}]=1-p_e,\label{eq-cutoff_trans_prob4}\\
				&\Pr[X_e(t+1)=1,M_e(t+1)=m+1|X_e(t)=1,M_e(t)=m]=1,\,\, 0\leq m\leq t^{\star}-1.\label{eq-cutoff_trans_prob5}
			\end{align}
			For $t^{\star}=0$, we only have the transition probabilities in \eqref{eq-cutoff_trans_prob1}--\eqref{eq-cutoff_trans_prob4}. Since these transition probabilities are time independent, and because the pair $(X_e(t+1),M_e(t+1))$ depends only on $(X_e(t),M_e(t))$, we have that $((X_e(t),M_e(t)):t\geq 1)$ is a stationary/time-homogeneous Markov process. As such, the conditional probabilities can be organized into the transition matrix $T_e(t^{\star})$, $t^{\star}\in\mathbbm{N}_0$, defined as follows:
			\begin{multline}\label{eq-cutoff_trans_prob6}
				\left(T_e(t^{\star})\right)_{\substack{x,m\\x',m'}}\coloneqq \Pr[X_e(t+1)=x,M_e(t+1)=m|X_e(t)=x',M_e(t)=m'],\\x,x'\in\{0,1\},~m,m'\in\{-1,0,1,\dotsc,t^{\star}\},
			\end{multline}
			with the matrix entries corresponding to the conditional probabilities not in \eqref{eq-cutoff_trans_prob1}--\eqref{eq-cutoff_trans_prob5} being set to zero. It is easy to verify that the results of Theorem~\ref{thm-avg_mem_status_inf} correspond to the stationary distribution of the Markov process defined by the transition matrices $T_e(t^{\star})$, $t^{\star}\in\mathbb{N}_0$. In other words, if we let $\vec{p}(t^{\star})$ be a column vector with elements given by
			\begin{equation}
				p_{x,m}(t^{\star})\coloneqq\lim_{t\to\infty}\Pr[M_e(t)=m,X_e(t)=x],
			\end{equation}
			for all $m\in\{-1,0,1,\dotsc,t^{\star}\}$ and $x\in\{0,1\}$, then it is straightforward to show that 
			\begin{equation}
				T_e(t^{\star})\vec{p}(t^{\star})=\vec{p}(t^{\star}).
			\end{equation}
			
			For $t^{\star}=\infty$, using \eqref{eq-mem_cutoff_infty_dt_simpler}, the transition probabilities in \eqref{eq-link_trans_prob_1}--\eqref{eq-link_trans_prob_4} can be simplified to the following:
			\begin{align}
				\Pr[X_e(t+1)=0|X_e(t)=0]&=1-p_e,\\
				\Pr[X_e(t+1)=1|X_e(t)=0]&=p_e,\\
				\Pr[X_e(t+1)=0|X_e(t)=1]&=0,\\
				\Pr[X_e(t+1)=1|X_e(t)=1]&=1.
			\end{align}
			These transition probabilities are time independent and Markovian, so they can be organized into the transition matrix $T_e(\infty)$ defined as follows:
			\begin{equation}\label{eq-cutoff_trans_infty}
				\left(T_e(\infty)\right)_{\substack{x\\x'}}\coloneqq \Pr[X_e(t+1)=x|X_e(t)=x'],\quad x,x'\in\{0,1\}.
			\end{equation}
	\end{itemize}
	
	\begin{table}
		\centering
		{\scriptsize\begin{tabular}{|c c c c c c c c c c||c|c|c|c|c}
			\hline $x_1$ & $x_2$ & $x_3$ & $x_4$ & $x_5$ & $x_6$ & $x_7$ & $x_8$ & $x_9$ & $x_{10}$ & $Y_e^{t^{\star}}\!(t)(h^t)$ & $Z_e^{t^{\star}}\!(t)(h^t)$ & $\Pr[H_e(t)=h^t]_{t^{\star}}$ & $M_e(t)(h^t)$ \\ \hline\hline
			0 & 0 & 0 & 0 & 0 & 0 & 0 & 0 & 0 & 0 & 0 & 0 & $(1-p_e)^{10}$ & $-1$  \\ \hline
			1 & 1 & 1 & 1 & 0 & 0 & 0 & 0 & 0 & 0 & 1 & 0 & $p_e(1-p_e)^6$ & $-1$  \\
			0 & 1 & 1 & 1 & 1 & 0 & 0 & 0 & 0 & 0 & 1 & 0 & $p_e(1-p_e)^6$ & $-1$  \\
			0 & 0 & 1 & 1 & 1 & 1 & 0 & 0 & 0 & 0 & 1 & 0 & $p_e(1-p_e)^6$ & $-1$  \\
			0 & 0 & 0 & 1 & 1 & 1 & 1 & 0 & 0 & 0 & 1 & 0 & $p_e(1-p_e)^6$ & $-1$  \\
			0 & 0 & 0 & 0 & 1 & 1 & 1 & 1 & 0 & 0 & 1 & 0 & $p_e(1-p_e)^6$ & $-1$  \\
			0 & 0 & 0 & 0 & 0 & 1 & 1 & 1 & 1 & 0 & 1 & 0 & $p_e(1-p_e)^6$ & $-1$  \\ \hline
			1 & 1 & 1 & 1 & 1 & 1 & 1 & 1 & 0 & 0 & 2 & 0 & $p_e^2(1-p_e)^2$ & $-1$  \\
			1 & 1 & 1 & 1 & 0 & 1 & 1 & 1 & 1 & 0 & 2 & 0 & $p_e^2(1-p_e)^2$ & $-1$  \\
			0 & 1 & 1 & 1 & 1 & 1 & 1 & 1 & 1 & 0 & 2 & 0 & $p_e^2(1-p_e)^2$ & $-1$  \\ \hline
			0 & 0 & 0 & 0 & 0 & 0 & 0 & 0 & 0 & 1 & 0 & 1 & $p_e(1-p_e)^{9}$ & 0  \\
			0 & 0 & 0 & 0 & 0 & 0 & 0 & 0 & 1 & 1 & 0 & 2 & $p_e(1-p_e)^{8}$ & 1  \\
			0 & 0 & 0 & 0 & 0 & 0 & 0 & 1 & 1 & 1 & 0 & 3 & $p_e(1-p_e)^{7}$ & 2  \\
			0 & 0 & 0 & 0 & 0 & 0 & 1 & 1 & 1 & 1 & 0 & 4 & $p_e(1-p_e)^{6}$ & 3  \\ \hline
			1 & 1 & 1 & 1 & 0 & 0 & 0 & 0 & 0 & 1 & 1 & 1 & $p_e^2(1-p_e)^5$ & 0  \\
			0 & 1 & 1 & 1 & 1 & 0 & 0 & 0 & 0 & 1 & 1 & 1 & $p_e^2(1-p_e)^5$ & 0  \\
			0 & 0 & 1 & 1 & 1 & 1 & 0 & 0 & 0 & 1 & 1 & 1 & $p_e^2(1-p_e)^5$ & 0  \\
			0 & 0 & 0 & 1 & 1 & 1 & 1 & 0 & 0 & 1 & 1 & 1 & $p_e^2(1-p_e)^5$ & 0 \\
			0 & 0 & 0 & 0 & 1 & 1 & 1 & 1 & 0 & 1 & 1 & 1 & $p_e^2(1-p_e)^5$ & 0  \\
			0 & 0 & 0 & 0 & 0 & 1 & 1 & 1 & 1 & 1 & 1 & 1 & $p_e^2(1-p_e)^5$ & 0  \\ \hline
			1 & 1 & 1 & 1 & 0 & 0 & 0 & 0 & 1 & 1 & 1 & 2 & $p_e^2(1-p_e)^4$ & 1  \\
			0 & 1 & 1 & 1 & 1 & 0 & 0 & 0 & 1 & 1 & 1 & 2 & $p_e^2(1-p_e)^4$ & 1  \\
			0 & 0 & 1 & 1 & 1 & 1 & 0 & 0 & 1 & 1 & 1 & 2 & $p_e^2(1-p_e)^4$ & 1  \\
			0 & 0 & 0 & 1 & 1 & 1 & 1 & 0 & 1 & 1 & 1 & 2 & $p_e^2(1-p_e)^4$ & 1  \\
			0 & 0 & 0 & 0 & 1 & 1 & 1 & 1 & 1 & 1 & 1 & 2 & $p_e^2(1-p_e)^4$ & 1  \\ \hline
			1 & 1 & 1 & 1 & 0 & 0 & 0 & 1 & 1 & 1 & 1 & 3 & $p_e^2(1-p_e)^3$ & 2  \\
			0 & 1 & 1 & 1 & 1 & 0 & 0 & 1 & 1 & 1 & 1 & 3 & $p_e^2(1-p_e)^3$ & 2  \\
			0 & 0 & 1 & 1 & 1 & 1 & 0 & 1 & 1 & 1 & 1 & 3 & $p_e^2(1-p_e)^3$ & 2  \\
			0 & 0 & 0 & 1 & 1 & 1 & 1 & 1 & 1 & 1 & 1 & 3 & $p_e^2(1-p_e)^3$ & 2  \\ \hline
			1 & 1 & 1 & 1 & 0 & 0 & 1 & 1 & 1 & 1 & 1 & 4 & $p_e^2(1-p_e)^2$ & 3  \\
			0 & 1 & 1 & 1 & 1 & 0 & 1 & 1 & 1 & 1 & 1 & 4 & $p_e^2(1-p_e)^2$ & 3  \\
			0 & 0 & 1 & 1 & 1 & 1 & 1 & 1 & 1 & 1 & 1 & 4 & $p_e^2(1-p_e)^2$ & 3  \\ \hline
			1 & 1 & 1 & 1 & 1 & 1 & 1 & 1 & 0 & 1 & 2 & 1 & $p_e^3(1-p_e)$ & 0  \\
			1 & 1 & 1 & 1 & 0 & 1 & 1 & 1 & 1 & 1 & 2 & 1 & $p_e^3(1-p_e)$ & 0  \\
			0 & 1 & 1 & 1 & 1 & 1 & 1 & 1 & 1 & 1 & 2 & 1 & $p_e^3(1-p_e)$ & 0 \\ \hline
			1 & 1 & 1 & 1 & 1 & 1 & 1 & 1 & 1 & 1 & 2 & 2 & $p_e^3$ & 1 \\ \hline
		\end{tabular}}
		\caption{Elementary link status sequences $(x_1,x_2,\dotsc,x_{10})$ for an elementary link in a quantum network, specified by an edge $e$ of the corresponding graph $G$, with $t^{\star}=3$ up to time $t=10$. The quantity $Y_e^{t^{\star}}\!(t)$ is the number of full blocks of ones in the elementary link status sequence up to time $t-1$, and $Z_e^{t^{\star}}\!(t)$ is the number of trailing ones in the elementary link status sequence up to time $t$. $M_e(t)$ is the memory time at time $t$, given by the formula in \eqref{eq-mem_time_def2a}.}\label{table-link_cutoff_example}
	\end{table}
	
	Next, let us consider what the histories $h^t$ look like through an particular example. Consider an elementary link for which $t^{\star}=3$, and let us consider the status of the elementary link up to time $t=10$. Given that each elementary link request succeeds with probability $p_e$ and fails with probability $1-p_e$, in Table~\ref{table-link_cutoff_example} we write down the probability for each sequence of elementary link statuses according to the formula in \eqref{eq-hist_prob_general}. Note that we only include those histories that have non-zero probability (indeed, some sequences $h^t=(x_1,a_1,\dotsc,a_{t-1},x_t)\in\{0,1\}^{2t-1}$ have zero probability with respect to the memory-cutoff policy). We also include in the table the memory times $M_e(t)$, which are calculated using the formula in \eqref{eq-mem_time_def2a}. Since the memory-cutoff policy is deterministic, it suffices to keep track only of the elementary link statuses $(x_1,\dotsc,x_t)$ and not the action values, because the action values are given deterministically by the elementary link statuses. For the elementary link status sequences, we define two quantities that are helpful for obtaining analytic formulas for the figures of merit defined in Section~\ref{sec-figures_of_merit_general}. The first quantity is $Y_e^{t^{\star}}\!(t)$, which we define to be the number of full blocks of ones (having length $t^{\star}+1$) in elementary link status sequences up to time $t-1$. The values that $Y_e^{t^{\star}}\!(t)$ can take are $0,1,\dotsc,\floor{\frac{t-1}{t^{\star}+1}}$ if $t^{\star}<\infty$, and 0 if $t^{\star}=\infty$. We also define the quantity $Z_e^{t^{\star}}\!(t)$ to be the number of trailing ones in elementary link status sequences up to time $t$. The values that $Z_e^{t^{\star}}\!(t)$ can take are $0,1,\dotsc,t^{\star}+1$ if $t^{\star}<\infty$, and $0,1,\dotsc,t$ if $t^{\star}=\infty$.
	
	Using the random variables $Y_e^{t^{\star}}\!(t)$ and $Z_e^{t^{\star}}\!(t)$, along with the general formula in \eqref{eq-hist_prob_general}, we obtain the following formula for the probability of histories with non-zero probability.
	
	\begin{proposition}\label{prop-time_seq_prob}
		For every time $t\geq 1$, cutoff $t^{\star}\in\mathbb{N}_0$, success probability $p_e\in[0,1]$, and history $h^t=(x_1,a_1,x_2,\allowbreak a_2,\dotsc,a_{t-1},x_t)\in\{0,1\}^{2t-1}$ with non-zero probability,
		\begin{multline}\label{eq-time_seq_prob}
			\Pr[H_e(t)=h^t]_{t^{\star}}=p_e^{Y_e^{t^{\star}}\!(t)(h^t)}(1-p_e)^{t-(t^{\star}+1)Y_e^{t^{\star}}\!(t)(h^t)}\delta_{Z_e^{t^{\star}}\!(t)(h^t),0}\\+\left(1-\delta_{Z_e^{t^{\star}}\!(t)(h^t),0}\right)p_e^{Y_e^{t^{\star}}\!(t)(h^t)+1}(1-p_e)^{t-Z_e^{t^{\star}}\!(t)(h^t)-(t^{\star}+1)Y_e^{t^{\star}}\!(t)(h^t)},
		\end{multline}
		where $Y_e^{t^{\star}}\!(t)(h^t)$ is defined to be the number of full blocks of ones of length $t^{\star}+1$ up to time $t-1$ in the sequence $(x_1,x_2,\dotsc,x_t)$ of elementary link statuses, and $Z_e^{t^{\star}}\!(t)(h^t)$ is defined to be the number of trailing ones in the sequence $(x_1,x_2,\dotsc,x_t)$. For $t^{\star}=\infty$,
		\begin{equation}\label{eq-time_seq_prob_infty}
			\Pr[H_e(t)=h^t]_{\infty}=(1-p_e)^t\delta_{Z_e^{\infty}(t)(h^t),0}+\left(1-\delta_{Z_e^{\infty}(t)(h^t),0}\right)p_e(1-p_e)^{t-Z_e^{\infty}(t)(h^t)}.~\defqedspec
		\end{equation}
	\end{proposition}
	
	\begin{proof}
		The result in \eqref{eq-time_seq_prob} follows immediately from the formula in \eqref{eq-hist_prob_general} by observing that $N_e^{\text{succ}}(t)=Y_e^{t^{\star}}\!(t)+1-\delta_{Z_e^{t^{\star}}\!(t),0}$ and $N_e^{\text{req}}(t)=t-(t^{\star}+1)Y_e^{t^{\star}}\!(t)-Z_e^{t^{\star}}\!(t)$. For $t^{\star}=\infty$, we always only have trailing ones in the elementary link status sequences, so that $Y_e^{\infty}(t)(h^t)=0$ for all $t\geq 1$ and every history $h^t$. The result in \eqref{eq-time_seq_prob_infty} then follows.
	\end{proof}
	
	Next, let us count the number of elementary link status sequences with non-zero probability. Using Table~\ref{table-link_cutoff_example} as a guide, we obtain the following.
	
	\begin{lemma}\label{prop-num_time_seq}
		For every time $t\geq 1$ and every cutoff $t^{\star}\in\mathbb{N}_0\cup\{\infty\}$, let $\Xi(t;t^{\star})$ denote the set of elementary link status sequences for the $t^{\star}$ memory-cutoff policy that have non-zero probability. Then, for $t^{\star}\in\mathbb{N}_0$, the number of elements in the set $\Xi(t;t^{\star})$ is
		\begin{equation}\label{eq-num_time_seq}
			\abs{\Xi(t;t^{\star})}=\sum_{x=0}^{\floor{\frac{t-1}{t^{\star}+1}}}\sum_{k=0}^{t^{\star}+1}\left(\binom{t-1-xt^{\star}}{x}\delta_{k,0}+(1-\delta_{k,0})\binom{t-k-xt^{\star}}{x}\boldsymbol{1}_{t-k-x(t^{\star}+1)\geq 0}\right),
		\end{equation}
		where
		\begin{equation}
			\boldsymbol{1}_{t-k-x(t^{\star}+1)\geq 0}=\left\{\begin{array}{l l} 1 & \text{if }t-k-x(t^{\star}+1)\geq 0,\\0 & \text{otherwise}. \end{array}\right.
		\end{equation}
		For $t^{\star}=\infty$, $\abs{\Xi(t;\infty)}=1+t$.~\defqed
	\end{lemma}
	
	\begin{proof}
		We start by counting the number of elementary link status sequences when the number of trailing ones is equal to zero, i.e., when $k\equiv Z_e^{t^{\star}}\!(t)(h^t)=0$. If we also let the number $x\equiv Y_e^{t^{\star}}\!(t)(h^t)$ of full blocks of ones in time $t-1$ be equal to one, then there are $t^{\star}+1$ ones and $t-t^{\star}-2$ zeros up to time $t-1$. The total number of elementary link status sequences is then equal to the number of ways that the single block of ones can be moved around in the elementary link status sequence up to time $t-1$. This quantity is equivalent to the number of permutations of $t-1-t^{\star}$ objects with $t-t^{\star}-2$ of them being identical (these are the zeros), which is given by
		\begin{equation}
			\frac{(t-1-t^{\star})!}{(t-2-t^{\star})!(t-1-t^{\star}-t+t^{\star}+2)!}=\frac{(t-1-t^{\star})!}{(t-t^{\star}-2)!(1)!}=\binom{t-1-t^{\star}}{1}.
		\end{equation}
		We thus have the $x=0$ and $k=0$ term in the sum in \eqref{eq-num_time_seq}. If we stick to $k=0$ but now consider more than one full block of ones in time $t-1$ (i.e., let $x\equiv Y_e^{t^{\star}}\!(t)(h^t)\geq 1$), then the number of elementary link status sequences is given by a similar argument as before: it is equal to the number of ways of permuting $t-1-xt^{\star}$ objects, with $x$ of them being identical (the blocks of ones) and the remaining $t-1-x(t^{\star}+1)$ objects also identical (the number of zeros), i.e., $\binom{t-1-xt^{\star}}{x}$. The total number of elementary link status sequences with zero trailing ones is therefore
		\begin{equation}\label{eq-num_time_seq_pf1}
			\sum_{x=0}^{\floor{\frac{t-1}{t^{\star}+1}}}\binom{t-1-xt^{\star}}{x}.
		\end{equation}
		
		Let us now consider the case $k\equiv Z_e^{t^{\star}}\!(t)(h^t)>0$. Then, the number of time slots in which full blocks of ones can be shuffled around is $t-k$. If there are $x$ blocks of ones in time $t-k$, then by the same arguments as before, the number of such elementary link status sequences is given by the number of ways of permuting $t-k-xt^{\star}$ objects, with $x$ of them being identical (the full blocks of ones) and the remaining $t-k-x(t^{\star}+1)$ of them also identical (these are the zeros up to time $t-k$). In other words, the number of elementary link status sequences with $k>0$ and $x\geq 0$ is
		\begin{equation}\label{eq-num_time_seq_pf2}
			\binom{t-k-xt^{\star}}{x}\boldsymbol{1}_{t-k-x(t^{\star}+1)\geq 0}.
		\end{equation}
		We must put the indicator function $\boldsymbol{1}_{t-k-x(t^{\star}+1)\geq 0}$ in order to ensure that the binomial coefficient makes sense. This also means that, depending on the time $t$, not all values of $k$ between 0 and $t^{\star}+1$ can be considered in the total number of elementary link status sequences (simply because it might not be possible to fit all possible values of trailing ones and full blocks of ones within that amount of time). By combining \eqref{eq-num_time_seq_pf1} and \eqref{eq-num_time_seq_pf2}, we obtain the desired result.
		
		In the case $t^{\star}=\infty$, because there are never any full blocks of ones and only trailing ones, we have $t$ elementary link status sequences, each containing $k$ trailing ones, where $1\leq k\leq t$. We also have an elementary link status sequence consisting of all zeros, giving a total of $t+1$ elementary link status sequences.
	\end{proof}
	\smallskip
	\begin{remark}
		Note that when $t^{\star}=0$, we get
		\begin{align}
			\abs{\Xi(t;0)}&=\sum_{x=0}^{t-1}\sum_{k=0}^1\left(\binom{t-1}{x}\delta_{k,0}+(1-\delta_{k,0})\binom{t-k}{x}\boldsymbol{1}_{t-k-x\geq 0}\right)\\
			&=\sum_{x=0}^{t-1}\binom{t-1}{x}+\sum_{x=0}^{t-1}\binom{t-1}{x}\underbrace{\boldsymbol{1}_{t-1-x\geq 0}}_{1~\forall x}\\
			&=2^{t-1}+2^{t-1}\\
			&=2^t.
		\end{align}
		In other words, when $t^{\star}=0$, \textit{all} $t$-bit strings are valid elementary link status sequences. 
		
		For $t\leq t^{\star}+1$, no full blocks of ones in time $t-1$ are possible, so we get
		\begin{align}
			\abs{\Xi(t;t^{\star})}&=\sum_{k=0}^{t^{\star}+1}\left(\binom{t-1}{0}\delta_{k,0}+(1-\delta_{k,0})\binom{t-k}{0}\boldsymbol{1}_{t-k\geq 0}\right)\\
			&=\binom{t-1}{0}+\sum_{k=1}^t \binom{t-k}{0}\\
			&=1+t.
		\end{align}
		This coincides with the result for $t^{\star}=\infty$, because when $t^{\star}=\infty$ the condition $t\leq t^{\star}+1$ is satisfied for all $t\geq 1$.~\defqed
	\end{remark}

\subsection{Proof of Theorem~\ref{thm-mem_status_pr}}\label{sec-mem_status_pr_pf}

	We start with the proof of the claimed expressions for $\Pr[M_e(t)=m,X_e(t)=1]_{t^{\star}}$. For $t\leq t^{\star}+1$, because no full blocks of ones up to time $t-1$ are possible, the possible values for the memory time are $0,1,\dotsc,t-1$. Furthermore, for each value of $m\in\{0,1,\dotsc,t-1\}$, there is only one elementary link status sequence for which $M_e(t)=m$, and this sequence has $Z_e^{t^{\star}}\!(t)=m+1$ trailing ones and thus probability $p_e(1-p_e)^{t-1-m}$ by Proposition~\ref{prop-time_seq_prob}.
		
	For $t>t^{\star}+1$, the possible values of the memory time are $0,1,\dotsc,t^{\star}$. Consider the number $Y_e^{t^{\star}}\!(t)$ of full blocks of ones in time $t-1$ and the number $Z_e^{t^{\star}}\!(t)$ of trailing ones in elementary link status sequences $(x_1,x_2,\dotsc,x_t)$ such that $x_t=1$. Since we must have $x_t=1$, we require $Z_e^{t^{\star}}\!(t)\geq 1$. Now, in order to have a memory time of $M_e(t)=m\in\{0,1,\dotsc,t^{\star}\}$, we can have elementary link status sequences consisting of any number $x=Y_e^{t^{\star}}\!(t)$ of full blocks of ones ranging from 0 to $\floor{\frac{t-1}{t^{\star}+1}}$ as long as $Z_e^{t^{\star}}\!(t)=m+1$. (Note that at the end of each full block of ones the memory time is equal to $t^{\star}$.) The number of such elementary link status sequences is
	\begin{equation}
		\binom{t-(m+1)-xt^{\star}}{x}\boldsymbol{1}_{t-(m+1)-x(t^{\star}+1)\geq 0},
	\end{equation}
	as given by \eqref{eq-num_time_seq_pf2}, and the probability of each such elementary link status sequence is $p_e^{x+1}(1-p_e)^{t-(m+1)-x(t^{\star}+1)}$. By summing with respect to all $0\leq x\leq\floor{\frac{t-1}{t^{\star}+1}}$, we obtain the desired result.

	We now prove the claimed expressions for $\Pr[M_e(t)=m,X_e(t)=0]_{t^{\star}}$. For finite $t^{\star}$, when $t\leq t^{\star}+1$, there is only one elementary link status sequence ending with a zero, and that is the sequence consisting of all zeros, which has probability $(1-p_e)^t$. Furthermore, since the value of the memory for this sequence is equal to $-1$, only the case $M_e(t)=-1$ has non-zero probability. When $t>t^{\star}+1$, we can again have non-zero probability only for $M_e(t)=-1$. In this case, because every elementary link status sequence has to end with a zero, we must have $Z_e^{t^{\star}}\!(t)=0$. Therefore, using \eqref{eq-time_seq_prob}, along with \eqref{eq-num_time_seq_pf1}, we obtain the desired result.
		
	For $t^{\star}=\infty$, only the elementary link status sequence consisting of all zeros ends with a zero, and in this case we have $M_e(t)=-1$. The result then follows.

\subsection{Proof of Theorem~\ref{thm-avg_mem_status_inf}}\label{app-avg_mem_status_inf_pf}

	We start by proving \eqref{eq-link_status_avg_inf}. Since we consider the limit $t\to\infty$, it suffices to consider the expression for $\Pr[X_e(t)=1]_{t^{\star}}$ in \eqref{eq-link_status_Pr1} for $t>t^{\star}+1$. Also due to the $t\to\infty$ limit, we can disregard the indicator function in \eqref{eq-link_status_Pr1}, so that
	\begin{equation}\label{eq-link_status_avg_inf_pf0}
		\lim_{t\to\infty}\mathbb{E}[X_e(t)]_{t^{\star}}=\lim_{t\to\infty}\sum_{x=0}^{\floor{\frac{t-1}{t^{\star}+1}}}\sum_{k=1}^{t^{\star}+1}\binom{t-k-xt^{\star}}{x}p_e^{x+1}(1-p_e)^{t-k-(t^{\star}+1)x}.
	\end{equation}
	Next, consider the binomial expansion of $(1-p_e)^{t-k-(t^{\star}+1)x}$:
	\begin{equation}
		(1-p_e)^{t-k-(t^{\star}+1)x}=\sum_{j=0}^{\infty}\binom{t-k-(t^{\star}+1)x}{j}(-1)^j p_e^j.
	\end{equation}
	Substituting this into \eqref{eq-link_status_avg_inf_pf0} gives us
	\begin{align}
		\lim_{t\to\infty}\mathbb{E}[X_e(t)]_{t^{\star}}&=p_e\lim_{t\to\infty}\sum_{x,j=0}^{\infty}\sum_{k=1}^{t^{\star}+1}\binom{t-k-t^{\star}x}{x}\binom{t-k-(t^{\star}+1)x}{j}(-1)^j p_e^{x+j}\\
		&=p_e\lim_{t\to\infty}\sum_{\ell=0}^{\infty}\sum_{j=0}^{\ell}\sum_{k=1}^{t^{\star}+1}\binom{t-k-t^{\star}j}{j}\binom{t-k-(t^{\star}+1)j}{\ell-j}(-1)^{\ell-j} p_e^{\ell}.\label{eq-link_status_avg_inf_pf4}
	\end{align}
	Now, for brevity, let $a\equiv t-k$, and let us focus on the sum
	\begin{equation}\label{eq-link_status_avg_inf_pf1}
		\sum_{j=0}^{\ell}(-1)^{\ell-j}\binom{a-t^{\star}j}{j}\binom{a-t^{\star}j-j}{\ell-j}.
	\end{equation}
	We start by expanding the binomial coefficients to get
	\begin{align}
		\binom{a-t^{\star}j}{j}\binom{a-t^{\star}j-j}{\ell-j}&=\frac{(a-t^{\star}j)!}{j!(\ell-j)!(a-t^{\star}j-\ell)!}\\
		&=\frac{1}{j!(\ell-j)!}\prod_{s=0}^{\ell-1}(a-t^{\star}j-s)\\
		&=\frac{1}{\ell!}\binom{\ell}{j}\prod_{s=0}^{\ell-1}(a-t^{\star}j-s).
	\end{align}
	Next, we have
	\begin{equation}
		\prod_{s=0}^{\ell-1}(a-t^{\star}j-s)=\sum_{n=0}^{\ell}(-1)^{\ell-n}\begin{bmatrix}\ell\\n\end{bmatrix}(a-t^{\star}j)^n,
	\end{equation}
	where $\begin{bmatrix}\ell\\n\end{bmatrix}$ is the (unsigned) Stirling number of the first kind\footnote{This number is defined to be the number of permutations of $\ell$ elements with $n$ disjoint cycles.}. Performing the binomial expansion of $(a-t^{\star}j)^n$, the sum in \eqref{eq-link_status_avg_inf_pf1} becomes
	\begin{equation}\label{eq-link_status_avg_inf_pf3}
		\sum_{j=0}^{\ell}\sum_{n=0}^{\ell}\sum_{i=0}^n (-1)^{\ell-j}\frac{1}{\ell!}\binom{\ell}{j}\begin{bmatrix}\ell\\n\end{bmatrix}\binom{n}{i}(-1)^i(t^{\star})^i j^i a^{n-i}.
	\end{equation}
	Now, it holds that
	\begin{equation}\label{eq-link_status_avg_inf_pf2}
		\sum_{j=0}^{\ell}(-1)^{\ell-j}\frac{1}{\ell!}\binom{\ell}{j}j^i=(-1)^{2\ell}\begin{Bmatrix}i\\\ell\end{Bmatrix},
	\end{equation}
	where $\begin{Bmatrix}i\\\ell\end{Bmatrix}$ is the Stirling number of the second kind\footnote{This number is defined to be the number of ways to partition a set of $i$ objects into $\ell$ non-empty subsets.}. For $i<\ell$, it holds that $\begin{Bmatrix}i\\\ell\end{Bmatrix}=0$, and $\begin{Bmatrix}\ell\\\ell\end{Bmatrix}=1$. Since $i$ ranges from 0 to $n$, and $n$ itself ranges from 0 to $\ell$, the sum in \eqref{eq-link_status_avg_inf_pf2} is zero except for when $i=\ell$. The sum in \eqref{eq-link_status_avg_inf_pf2} is therefore effectively equal to $(-1)^{2\ell}\delta_{i,\ell}$. Substituting this into \eqref{eq-link_status_avg_inf_pf3} leads to
	\begin{equation}
		\sum_{n=0}^{\ell}\sum_{i=0}^n(-1)^{2\ell}\delta_{i,\ell}\begin{bmatrix}\ell\\n\end{bmatrix}\binom{n}{i}(-1)^i(t^{\star})^i a^{n-i}=(-1)^{\ell}(t^{\star})^{\ell},
	\end{equation}
	where we have used the fact that $\begin{bmatrix}\ell\\\ell\end{bmatrix}=1$. Altogether, we have shown that
	\begin{equation}\label{eq-special_sum}
		\sum_{j=0}^{\ell}(-1)^{\ell-j}\binom{a-t^{\star}j}{j}\binom{a-t^{\star}j-j}{\ell-j}=(-1)^{\ell}(t^{\star})^{\ell}
	\end{equation}
	for all $\ell\geq 0$. The sum is independent of $a=t-k$. Substituting this result into \eqref{eq-link_status_avg_inf_pf4}, and using the fact that
	\begin{equation}\label{eq-homographic_func}
		\sum_{\ell=0}^{\infty} (-1)^{\ell} x^{\ell}=\frac{1}{1+x},\quad x\neq -1,
	\end{equation}
	we get
	\begin{equation}
		\lim_{t\to\infty}\mathbb{E}[X_e(t)]_{t^{\star}}=p_e\sum_{\ell=0}^{\infty}\sum_{k=1}^{t^{\star}+1}(-1)^{\ell}(t^{\star}p_e)^{\ell}=p_e(t^{\star}+1)\sum_{\ell=0}^{\infty} (-1)^{\ell} (t^{\star}p_e)^{\ell}=\frac{(t^{\star}+1)p_e}{1+t^{\star}p_e},
	\end{equation}
	as required.

	The proof of \eqref{eq-avg_mem_status_inf} is very similar to the proof of \eqref{eq-link_status_avg_inf}. Using the result of Theorem~\ref{thm-mem_status_pr}, in the limit $t\to\infty$ we have
	\begin{multline}
		\lim_{t\to\infty}\Pr[M_e(t)=m,X_e(t)=1]_{t^{\star}}\\=(1-\delta_{m,-1})\lim_{t\to\infty}\sum_{x=0}^{\infty} \binom{t-(m+1)-xt^{\star}}{x}p_e^{x+1}(1-p_e)^{t-(m+1)-x(t^{\star}+1)}.
	\end{multline}
	Using the binomial expansion of $(1-p_e)^{t-(m+1)-x(t^{\star}+1)}$, exactly as in the proof of \eqref{eq-link_status_avg_inf}, we can write
	\begin{align}
		&\lim_{t\to\infty}\Pr[M_e(t)=m,X_e(t)=1]_{t^{\star}}\nonumber\\
		&\qquad=(1-\delta_{m,-1})\lim_{t\to\infty}\sum_{x=0}^{\infty}\sum_{j=0}^{\infty} p_e\binom{t-(m+1)-xt^{\star}}{x}\binom{t-(m+1)-(t^{\star}+1)x}{j}(-1)^jp_e^{x+j}\\
		&\qquad=(1-\delta_{m,-1})\lim_{t\to\infty}\sum_{\ell=0}^{\infty}\sum_{j=0}^{\ell}p_e\binom{t-(m+1)-jt^{\star}}{j}\binom{t-(m+1)-(t^{\star}+1)j}{\ell-j}(-1)^{\ell-j}p_e^{\ell}.
	\end{align}
	Then, using \eqref{eq-special_sum}, we have that
	\begin{equation}
		\sum_{j=0}^{\ell}(-1)^{\ell-j}\binom{t-(m+1)-jt^{\star}}{j}\binom{t-(m+1)-(t^{\star}+1)j}{\ell-j} =(-1)^{\ell}(t^{\star})^{\ell}
	\end{equation}
	for all $t\geq 1$ and all $m\in\{0,1,\dotsc,t^{\star}\}$. Finally, using \eqref{eq-homographic_func}, we obtain
	\begin{equation}
		\lim_{t\to\infty}\Pr[M_e(t)=m,X_e(t)=1]_{t^{\star}}=p_e\sum_{\ell=0}^{\infty} (-1)^{\ell}(t^{\star}p_e)^{\ell}=\frac{p_e}{1+t^{\star}p_e},
	\end{equation}
	as required. The proof of \eqref{eq-mem_time_prob_infty_x0} is similar.

\subsection{Other figures of merit}\label{sec-other_figs_of_merit_mem_cutoff}

	We now consider the figures of merit defined in Section~\ref{sec-other_figs_of_merit} in the context of the memory-cutoff policy.

\subsubsection{Waiting time}\label{app-avg_waiting_time_req_pf}

	Let us consider the expected waiting time for an elementary link in a quantum network undergoing the memory-cutoff policy.
	
	\begin{theorem}\label{thm-avg_waiting_time_req}
		Consider an edge $e\in E$ in the graph $G=(V,E)$ corresponding to the elementary links of a quantum network, and let $p_e\in[0,1]$ be the success probability for the elementary link corresponding to $e$. For every cutoff $t^{\star}\in\mathbbm{N}_0$ and every request time $t_{\text{req}}\geq 0$, the expected waiting time for the elementary link is
		\begin{equation}\label{eq-avg_waiting_time_req}
			\mathbb{E}[W_{e}(t_{\text{req}})]_{t^{\star}}=\frac{\Pr[M_e(t_{\text{req}}+1)=-1,X_e(t_{\text{req}}+1)=0]_{t^{\star}}}{p_e(1-p_e)}.
		\end{equation}
		For $t^{\star}=\infty$,
		\begin{equation}
			\mathbb{E}[W_{e}(t_{\text{req}})]_{\infty}=\frac{\Pr[X_e(t_{\text{req}}+1)=0]_{\infty}}{p_e(1-p_e)}=\frac{(1-p_e)^{t_{\text{req}}}}{p_e}. \defqedspec
		\end{equation}
	\end{theorem}
	
	
	As a check, let us observe the following:
	\begin{itemize}
		\item If $t_{\text{req}}=0$, then because $\Pr[M_e(1)=-1,X_e(1)=0]_{t^{\star}}=1-p_e$ for all $t^{\star}\in\mathbbm{N}_0$ (see Theorem~\ref{thm-mem_status_pr}), we obtain $\mathbb{E}[W_e(0)]_{t^{\star}}=\frac{1}{p_e}$, as expected. We get the same result for $t^{\star}=\infty$.
		
		\item If $t^{\star}=0$, then we get $\Pr[M_e(t_{\text{req}}+1)=-1,X_e(t_{\text{req}}+1)=0]_{0}=1-p_e$ for all $t_{\text{req}}\geq 0$ (see Theorem~\ref{thm-mem_status_pr}), which means that $\mathbb{E}[W_{e}(t_{\text{req}})]_0=\frac{1}{p_e}$ for all $t_{\text{req}}\geq 0$. This makes sense, because in the $t^{\star}=0$ memory-cutoff policy the quantum state of the elementary link is never held in memory. 
	\end{itemize}

	\begin{proof}
		Using \eqref{eq-waiting_time_prob_late_request}, we have
		\begin{multline}
			\Pr[W_{e}(t_{\text{req}})=t]_{t^{\star}}=\Pr[X_e(t_{\text{req}}+1)=0,\dotsc,X_e(t_{\text{req}}+t)=1]_{t^{\star}}\\=\sum_{m_1,\dotsc,m_t=0}^{t^{\star}}\Pr[X_e(t_{\text{req}}+1)=0,M_e(t_{\text{req}}+1)=m_1,\dotsc,\\X_e(t_{\text{req}}+t)=1,M_e(t_{\text{req}}+t)=m_t]_{t^{\star}}.
		\end{multline}
		Using the transition matrix $T_e(t^{\star})$ defined in \eqref{eq-cutoff_trans_prob6}, we obtain
		\begin{multline}
			\Pr[W_{e}(t_{\text{req}})=t]_{t^{\star}}\\=\sum_{m_1,\dotsc,m_t=0}^{t^{\star}} (T_e(t^{\star}))_{\substack{1,m_t\\0,m_{t-1}}}\dotsb (T_e(t^{\star}))_{\substack{0,m_3\\0,m_2}}(T_e(t^{\star}))_{\substack{0,m_2\\0,m_1}}\\\times\Pr[M_e(t_{\text{req}}+1)=m_1,X_e(t_{\text{req}}+1)=0]_{t^{\star}}.
		\end{multline}
		Using \eqref{eq-mem_time_prob_x0}, along with \eqref{eq-cutoff_trans_prob1}--\eqref{eq-cutoff_trans_prob6}, we have that
		\begin{equation}
			\Pr[W_{e}(t_{\text{req}})=t]_{t^{\star}}=\Pr[M_e(t_{\text{req}}+1)=-1,X_e(t_{\text{req}}+1)=0]_{t^{\star}}p_e(1-p_e)^{t-2},
		\end{equation}
		for all $t\geq 1$. The result then follows.
		
		For $t^{\star}=\infty$, using the transition matrix $T_e(\infty)$ defined in \eqref{eq-cutoff_trans_infty} leads to
		\begin{multline}
			\Pr[X_e(t_{\text{req}}+1)=0,\dotsc,X_e(t_{\text{req}}+t)=1]_{\infty}\\ = (T_e(\infty))_{\substack{1\\0}}(T_e(\infty))_{\substack{0\\0}}\dotsb(T_e(\infty))_{\substack{0\\0}}\Pr[X_e(t_{\text{req}}+1)=0]_{\infty}.
		\end{multline}
		Then, from \eqref{eq-link_status_Pr1}, we have that $\Pr[X_e(t_{\text{req}}+1)=0]=(1-p_e)^{t_{\text{req}}+1}$, so that
		\begin{equation}
			\Pr[W_{e}(t_{\text{req}})=t]_{\infty}=p_e(1-p_e)^{t-2}(1-p_e)^{t_{\text{req}}+1}
		\end{equation}
		for all $t\geq 1$. The result then follows.
	\end{proof}
	
	\begin{figure}
		\centering
		\includegraphics[scale=1]{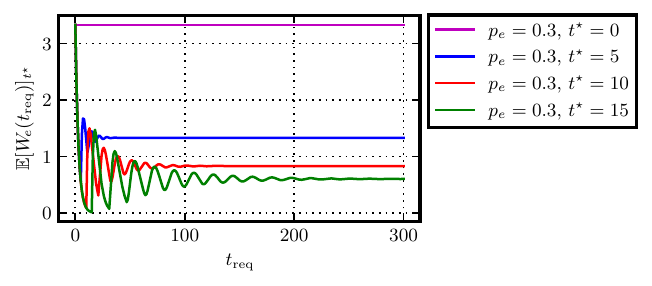}
		\caption{The expected waiting time for an elementary link, given by \eqref{eq-avg_waiting_time_req}, as a function of the request time $t_{\text{req}}$. We let the success probability for the elementary link be $p_e=0.3$, and we take various values for the cutoff $t^{\star}$.}\label{fig-avg_waiting_time_example}
	\end{figure}
	
	In the limit $t_{\text{req}}\to\infty$, we obtain using \eqref{eq-mem_time_prob_infty_x0},
	\begin{equation}\label{eq-wait_time_treqInfty}
		\lim_{t_{\text{req}}\to\infty}\mathbb{E}[W_{e}(t_{\text{req}})]_{t^{\star}}=\frac{1}{p_e(1+t^{\star}p_e)},\quad t^{\star}\in\mathbbm{N}_0.
	\end{equation}
	See Figure~\ref{fig-avg_waiting_time_example} for plots of the expected waiting time, given by \eqref{eq-avg_waiting_time_req}, as a function of the request time $t_{\text{req}}$ for various values of $t^{\star}$. As long as $t^{\star}$ is strictly greater than zero, the waiting time is strictly less than $\frac{1}{p_e}$, despite the oscillatory behavior for small values of $t_{\text{req}}$. In the limit $t_{\text{req}}\to\infty$, we see that the waiting time is monotonically decreasing with increasing $t^{\star}$, which is also apparent from \eqref{eq-wait_time_treqInfty}.

\subsubsection{Success rate}\label{app-thm-avg_sucess_rate_pf}

	Let us now consider the expected success rate for an elementary link undergoing the memory-cutoff policy.
	
	\begin{theorem}\label{thm-avg_sucess_rate}
		Consider an edge $e\in E$ in the graph $G=(V,E)$ corresponding to the elementary links of a quantum network, and let $p_e\in[0,1]$ be the success probability for the elementary link corresponding to $e$. For every cutoff $t^{\star}\in\mathbb{N}_0\cup\{\infty\}$ and every time $t\geq 1$, the expected success rate for the elementary link is
		\begin{equation}
			\mathbb{E}[S_{\!e}(t)]_{t^{\star}}=\sum_{j=0}^{t-1}\frac{1}{j+1}p_e(1-p_e)^j,\quad t\leq t^{\star}+1.
		\end{equation}
		For $t>t^{\star}+1$,
		\begin{multline}
			\mathbb{E}[S_{\!e}(t)]_{t^{\star}}=\sum_{x=0}^{\floor{\frac{t-1}{t^{\star}+1}}}\left(\frac{x}{t-t^{\star}x}\binom{t-1-xt^{\star}}{x}p_e^x(1-p_e)^{t-(t^{\star}+1)x}\right.\\\left.+\sum_{k=1}^{t^{\star}+1}\frac{x+1}{t-k-t^{\star}x+1}\binom{t-k-xt^{\star}}{x}p_e^{x+1}(1-p_e)^{t-k-(t^{\star}+1)x}\boldsymbol{1}_{t-k-(t^{\star}+1)x\geq 0}\right). \defqedspec
		\end{multline}
	\end{theorem}
	
	\begin{proof}
		We start with the observation that, for every history $h^t$, the number of successful requests can be written in terms of the number $Y_e^{t^{\star}}\!(t)(h^t)$ of blocks of ones of length $t^{\star}+1$ and the number $Z_e^{t^{\star}}\!(t)(h^t)$ of trailing ones in the elementary link status sequence corresponding to $h^t$ as
		\begin{equation}
			Y_e^{t^{\star}}\!(t)(h^t)+1-\delta_{Z_e^{t^{\star}}\!(t)(h^t),0}.
		\end{equation}
		Similarly, the total number of failed requests is
		\begin{equation}
			t-Z_e^{t^{\star}}\!(t)(h^t)-(t^{\star}+1)Y_e^{t^{\star}}\!(t)(h^t).
		\end{equation}
		Therefore,
		\begin{align}
			S_{\!e}(t)(h^t)&=\frac{Y_e^{t^{\star}}\!(t)(h^t)+1-\delta_{Z_e^{t^{\star}}\!(t)(h^t),0}}{t-Z_e^{t^{\star}}\!(t)(h^t)-(t^{\star}+1)Y_e^{t^{\star}}\!(t)(h^t)+Y_e^{t^{\star}}\!(t)(h^t)+1-\delta_{Z_e^{t^{\star}}\!(t)(h^t),0}}\\
			&=\frac{Y_e^{t^{\star}}\!(t)(h^t)+1-\delta_{Z_e^{t^{\star}}\!(t)(h^t),0}}{t-Z_e^{t^{\star}}\!(t)(h^t)-t^{\star}Y_e^{t^{\star}}\!(t)(h^t)+1-\delta_{Z_e^{t^{\star}}\!(t)(h^t),0}}.\label{eq-avg_success_rate_pf1}
		\end{align}
		Now, for $t\leq t^{\star}+1$, we always have $Y_e^{t^{\star}}\!(t)(h^t)=0$ for every history $h^t$, and the elementary link status sequence can consist only of a positive number of trailing ones not exceeding $t$. Thus, from Proposition~\ref{prop-time_seq_prob}, the probability of any such history is $p_e(1-p_e)^{t-Z_e^{t^{\star}}\!(t)(h^t)}$. Using \eqref{eq-avg_success_rate_pf1} then leads to
		\begin{align}
			\mathbb{E}[S_{\!e}(t)]_{t^{\star}}&=\sum_{h^t\in\{0,1\}^{2t-1}}S_{\!e}(t)(h^t)\Pr[H_e(t)=h^t]_{t^{\star}}\\
			&=\sum_{k=1}^t \frac{1}{t-k+1}p_e(1-p_e)^{t-k}\\
			&=\sum_{j=0}^{t-1}\frac{1}{j+1}p_e(1-p_e)^j
		\end{align}
		for $t\leq t^{\star}+1$, as required, where the last equality follows by a change of summation variable.
		
		For $t>t^{\star}+1$, we use \eqref{eq-avg_success_rate_pf1} again, keeping in mind this time that the number of trailing ones can be equal to zero, to get
		\begin{align}
			\mathbb{E}[S_{\!e}(t)]_{t^{\star}}&=\sum_{h^t\in\{0,1\}^{2t-1}}S_{\!e}(t)(h^t)\Pr[H_e(t)=h^t]_{t^{\star}}\\
			&=\sum_{h^t:Z_e^{t^{\star}}\!(t)(h^t)=0}S_{\!e}(t)(h^t)\Pr[H_e(t)=h^t]_{t^{\star}}\nonumber\\
			&\qquad\qquad\qquad\qquad+\sum_{h^t:Z_e^{t^{\star}}\!(t)(h^t)\geq 1}S_{\!e}(t)(h^t)\Pr[H_e(t)=h^t]_{t^{\star}}\\
			&=\sum_{x=0}^{\floor{\frac{t-1}{t^{\star}+1}}}\left(\frac{x}{t-t^{\star}x}\Pr[H_e(t)=h^t:Y_e^{t^{\star}}\!(t)(h^t)=x,Z_e^{t^{\star}}\!(t)(h^t)=0]_{t^{\star}}\right.\nonumber\\
			&\quad\left.+\sum_{k=1}^{t^{\star}+1}\frac{x+1}{t-k-t^{\star}x+1}\Pr[H_e(t)=h^t:Y_e^{t^{\star}}\!(t)(h^t)=x,Z_e^{t^{\star}}\!(t)(h^t)=k]_{t^{\star}}\right).
		\end{align}
		Using Proposition~\ref{prop-time_seq_prob}, we arrive at the desired result.
	\end{proof}
	
	\begin{figure}
		\centering
		\includegraphics[scale=1]{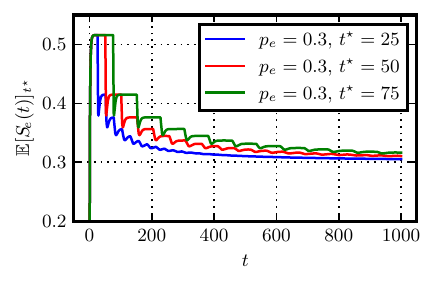}
		\caption{The expected success rate, as given by the expressions in Theorem~\ref{thm-avg_sucess_rate}, for an elementary link with success probability $p_e=0.3$ undergoing various memory-cutoff policies.}\label{fig-avg_success_rate_example}
	\end{figure}
	
	See Figure~\ref{fig-avg_success_rate_example} for a plot of the expected success rate $\mathbb{E}[S_{\!e}(t)]_{t^{\star}}$ as a function of time for various values of the cutoff $t^{\star}$. We find that the rate has essentially the shape of a decaying square wave, which is clearer for larger values of the cutoff. In particular, the ``plateaus'' in the curves have a period of $t^{\star}+1$ time steps. Consider the values of these pleateaus. The largest plateau can be found by considering the case $t^{\star}=\infty$, because in this case the condition $t\leq t^{\star}+1$ is satisfied for all $t\geq 1$, and it is when this condition is true that the largest plateau occurs. Using Theorem~\ref{thm-avg_sucess_rate} with $t^{\star}=\infty$, we find that the value of the largest plateau approaches
	\begin{equation}\label{eq-avg_succ_rate_tInfty_1}
		\lim_{t\to\infty}\mathbb{E}[S_{\!e}(t)]_{\infty}=\lim_{t\to\infty}\sum_{j=0}^{t-1}\frac{1}{j+1}p_e(1-p_e)^j=-\frac{p_e\ln p_e}{1-p_e},
	\end{equation}
	for all $p_e\in(0,1)$. In the case $t^{\star}\in\mathbb{N}_0$, as we see in Figure~\ref{fig-avg_success_rate_example}, there are multiple plateaus, with each plateau lasting for a period of $t^{\star}+1$ time steps, as mentioned earlier. The values of these pleateaus depend on the number $x\geq 0$ of full blocks of ones in the elementary link status sequence (see Appendix~\ref{app-mem_cutoff_details} for details). Specifically, the values of the plateaus approach
	\begin{multline}
		\lim_{t\to\infty}\sum_{k=1}^{t-(t^{\star}+1)x}\frac{x+1}{t-k-t^{\star}x+1}\binom{t-k-t^{\star}x}{x}p_e^{x+1}(1-p_e)^{t-k-(t^{\star}+1)x}\\=\lim_{t\to\infty}\sum_{j=(t^{\star}+1)x}^{t-1}\frac{x+1}{j-t^{\star}x+1}\binom{j-t^{\star}x}{x}p_e^{x+1}(1-p_e)^{j-(t^{\star}+1)x}=p_e\cdot{~}_2F_1(1,1,2+x,1-p_e),
	\end{multline}
	for all $x\geq 0$, where $_2F_1(a,b,c,z)$ is the hypergeometric function. Then, using the fact that $\lim_{x\to\infty} {~}_2F_1(1,1,2+x,1-p_e)=1$ \cite{CAJ17}, we conclude that the plateaus approach the value of $p_e$, i.e.,
	\begin{equation}\label{eq-avg_succ_rate_tInfty_2}
		\lim_{t\to\infty}\mathbb{E}[S_{\!e}(t)]_{t^{\star}}=p_e,\quad t^{\star}\in\mathbb{N}_0.
	\end{equation}

\end{appendices}

\newpage

\printbibliography[heading=bibintoc]

\end{document}